\documentclass{article}
\usepackage{hpcstyle}
\usepackage{amsthm} 
\usepackage{verbatim}
\usepackage{xcolor,svgcolor}
\usepackage{hyperref}
\usepackage{float}
\usepackage{stmaryrd}
\usepackage{cleveref}

\usepackage{amsmath,amssymb} 

\usepackage{url,xspace}
\usepackage{bussproofs}
\usepackage{mathbbol}
\usepackage[all]{xy}
\SelectTips{cm}{}
\usepackage{framed}

\usepackage{stmaryrd}

\newcommand{\oubli}[1]{} 

\newcommand{\trycatch}{\mathtt{try}\texttt{-}\mathtt{catch}}
\newcommand{\sto}{\!\!\to\!\!}
\newcommand{\successor}{\mathtt{s}}
\newcommand{\predecessor}{\mathtt{p}}

\title{\mytitle}

\author{Jean-Guillaume Dumas\thanks{Laboratoire J. Kuntzmann,
    Universit\'e de Grenoble. 51, rue des Math\'ematiques, umr CNRS
    5224, bp 53X, F38041 Grenoble, France,
    \href{mailto:Jean-Guillaume.Dumas@imag.fr,Dominique.Duval@imag.fr,Burak.Ekici@imag.fr}{\{Jean-Guillaume.Dumas,Dominique.Duval,Burak.Ekici\}@imag.fr}.}
  \and Dominique Duval\footnotemark[1]
  \and Burak Ekici\footnotemark[1]
  \and Damien~Pous\thanks{Plume team, CNRS, ENS Lyon, Universit\'e de
    Lyon, INRIA, UMR 5668, France, 
    \href{mailto:Damien.Pous@ens-lyon.fr}{Damien.Pous@ens-lyon.fr}.}
  \and Jean-Claude~Reynaud\thanks{Reynaud Consulting (RC), \href{mailto:Jean-Claude.Reynaud@imag.fr}{Jean-Claude.Reynaud@imag.fr}.}
}

\theoremstyle{plain}
\newtheorem{theorem}{Theorem}[section]
\newtheorem{proposition}[theorem]{Proposition}
\newtheorem{lemma}[theorem]{Lemma}
\theoremstyle{definition}
\newtheorem{definition}[theorem]{Definition}
\theoremstyle{remark}
\newtheorem{remark}[theorem]{Remark}
\newtheorem{example}[theorem]{Example}
\newtheorem{corollary}[theorem]{Corollary}
\newtheorem{assumption}[theorem]{Assumption}

\begin{document}
\maketitle
\begin{abstract}
A theory is complete if it does not contain a contradiction, while all
of its proper extensions do.
In this paper, first we introduce a relative notion of 
syntactic completeness; then we prove that adding exceptions 
to a programming language can be done in such a way that 
the completeness of the language is not made worse. 
These proofs are formalized in a logical system which is 
close to the usual syntax for exceptions,
and they have been checked with the proof assistant Coq.
\end{abstract}
\section{Introduction}
%
In computer science, 
an exception is an abnormal event occurring during the execution 
of a program. A mechanism for handling exceptions consists of 
two parts: an exception is \emph{raised} when an abnormal event occurs,
and it can be \emph{handled} later, by switching the execution to 
a specific subprogram. 
Such a mechanism is very helpful, but it is difficult 
for programmers to reason about it. 
A difficulty for reasoning about programs involving exceptions 
is that they are \emph{computational effects},
in the sense that their syntax does not look like their interpretation:%
typically,   
a piece of program with arguments in $X$ that returns a value in $Y$
is interpreted as a function from $X+E$ to $Y+E$ where $E$ is 
the set of exceptions. 
On the one hand, reasoning with $f:X\to Y$ is close to the syntax, 
but it is error-prone because it is not sound with respect to the semantics. 
On the other hand, reasoning with $f:X+E\to Y+E$ is sound but 
it loses most of the interest of the exception mechanism, 
where the propagation of exceptions is implicit: 
syntactically, $f:X\to Y$ may be followed by any $g:Y\to Z$, 
since the mechanism of exceptions will take care of propagating 
the exceptions raised by $f$, if any.
Another difficulty for reasoning about programs involving exceptions 
is that the handling mechanism is encapsulated in a $\trycatch$ block, 
while the behaviour of this mechanism is easier to explain 
in two parts (see for instance~\cite[Ch. 14]{java} for Java 
or~\cite[\S 15]{cpp} for C++): 
the $\catch$ part may recover from exceptions,
so that its interpretation may be any $f:X+E\to Y+E$, 
but the $\trycatch$ block must propagate exceptions,
so that its interpretation is determined by some $f:X\to Y+E$. 

In \cite{DDER14-exc} we defined a logical system 
for reasoning about states and exceptions and we used it for getting 
certified proofs of properties of programs in computer algebra, 
with an application to exact linear algebra. 
This logical system is called the \emph{decorated logic} 
for states and exceptions. Here we focus on exceptions.
The decorated logic for exceptions 
deals with $f:X\to Y$, without any mention of~$E$, 
however it is sound thanks to a classification of the terms and the equations. 
Terms are classified, as in a programming language, 
according to the way they may interact with exceptions: a term  
either has no interaction with exceptions (it is ``pure''), 
or it may raise exceptions and must propagate them, 
or it is allowed to catch exceptions 
(which may occur only inside the $\catch$ part of a $\trycatch$ block). 
The classification of equations follows a line that was introduced 
in \cite{DD10}: 
besides the usual ``strong'' equations, interpreted as equalities of functions, 
in the decorated logic for exceptions there are also ``weak'' equations, 
interpreted as equalities of functions on non-exceptional arguments.
This logic has been built so as to be sound, 
but little was known about its completeness. 
In this paper we prove a novel completeness result:
the decorated logic for exceptions is \emph{relatively Hilbert-Post complete},
which means that adding exceptions to a programming language can be done 
in such a way that the completeness of the language is not made worse. 
For this purpose, we first define and study the novel notion
of {\em relative} Hilbert-Post completeness, which seems to be 
a relevant notion for the completeness of various computational
effects:
indeed, we prove that this notion is preserved when combining
effects. 
Practically, this means that we have defined a decorated framework 
where reasoning about programs with and without exceptions are 
equivalent, in the following sense: 
if there exists an unprovable equation not contradicting the given 
decorated rules, then this equation is equivalent to a set of 
unprovable equations of the pure sublogic not contradicting its rules.

Informally, in classical logic, a consistent theory is one that does not
contain a contradiction and a theory is complete if it is
consistent, and none of its proper extensions is consistent.
Now, the usual (``\emph{absolute}'') Hilbert-Post completeness, 
also called Post completeness, is a syntactic notion of completeness 
which does not use any notion of negation,
so that it is well-suited for equational logic. 
In a given logic $L$, we call \emph{theory} a set of sentences 
which is deductively closed: everything you can derive from it
(using the rules of $L$) is already in it. 
Then, more formally, a theory is \emph{(Hilbert-Post) consistent} if it
does not contain  all sentences, 
and it is \emph{(Hilbert-Post) complete} if it is consistent 
and if any sentence which is added to it generates an inconsistent
theory~\cite[Def. 4]{Tarski30}.  

All our completeness proofs have been verified with the Coq proof
assistant.
First, this shows that it is possible to formally prove that programs
involving exceptions comply to their specifications.
Second, this is of help for improving the confidence in the
results.
Indeed, for a human prover, proofs in a decorated logic require some care: 
they look very much like familiar equational proofs, 
but the application of a rule may be subject to restrictions 
on the decoration of the premises of the rule. 
The use of a proof assistant in order to
check that these unusual restrictions were never violated has thus
proven to be quite useful.
Then, many of the proofs we give in this paper require a structural
induction. There, the correspondence between our proofs and their Coq
counterpart was eased, as structural induction is also at the core of
the design of Coq.

A major difficulty for reasoning about programs involving exceptions,
and more generally computational effects, 
is that their syntax does not look like their interpretation: typically,
a piece of program from $X$ to $Y$
is not interpreted as a function from $X$ to $Y$, because of the effects.
The best-known algebraic approach for dealing with this problem has been
initiated by Moggi: 
an effect is associated to a monad $T$, in such a way that 
the interpretation of a program from $X$ to $Y$ is a function
from $X$ to $T(Y)$~\cite{Moggi91}:
typically, for exceptions, $T(Y)=Y+E$. 
Other algebraic approaches include effect systems~\cite{LucassenGifford88}, 
Lawvere theories~\cite{PlotkinPower02}, 
algebraic handlers~\cite{PP09}, 
comonads~\cite{UV08,POM14}, 
dynamic logic~\cite{MSG10}, among others. 
Some completeness results have been obtained, for instance  
for (global) states~\cite{Pretnar10} and   
for local states~\cite{Staton10-fossacs}. 
The aim of these approaches is to extend functional languages 
with tools for programming and proving side-effecting programs;  
implementations include Haskell \cite{BHM00}, Idris \cite{Idris}, 
Eff \cite{BauerP15}, 
while Ynot \cite{Ynot} is a Coq library for writing and verifying 
imperative programs. 

Differently, our aim is to build a logical system for proving 
properties of some families of programs written in widely used 
non-functional languages like Java or C++\footnote{For instance, a
  denotational semantics of our framework for exceptions, which relies
  on the common semantics of exceptions in these languages, was given
  in~\cite[\S~4]{DDER14-exc}.}.
The salient features of our approach are that:
\\{\bf (1)}
The syntax of our logic is kept close to the syntax of 
programming languages. This is made possible by starting from a simple
syntax without effect and by adding decorations, which often
correspond to keywords of the languages, for taking the effects into
account. 
\\{\bf (2)}
We consider exceptions in two settings, the programming
  language and the core language. This enables for instance to
  separate the treatment, in proofs, of the matching between normal or
  exceptional  behavior from the actual recovery after an exceptional
  behavior. 

In Section~\ref{sec:hpc} we introduce a \emph{relative} notion of 
Hilbert-Post completeness in a logic $L$ with respect to a sublogic
$L_0$.  
Then in Section~\ref{sec:exc} we prove the relative Hilbert-Post completeness 
of a theory of exceptions based on the usual $\throw$ and 
$\trycatch$ statement constructors. 
We go further in Section~\ref{sec:excore} by establishing 
the relative Hilbert-Post completeness 
of a \emph{core} theory for exceptions with individualized 
$\trycore$ and $\catchcore$ statement constructors, 
which is useful for expressing the behaviour of the $\trycatch$
blocks. 
All our completeness proofs have been verified with the Coq proof
assistant and we therefore give the main ingredients of the framework
used for this verification and the correspondence between our Coq
package and the theorems and propositions of this paper in
Section~\ref{sec:coq}.

\section{Relative Hilbert-Post completeness}
\label{sec:hpc}

Each logic in this paper comes with a \emph{language}, 
which is a set of \emph{formulas}, 
and with \emph{deduction rules}.  
Deduction rules are used for deriving (or generating) \emph{theorems}, 
which are some formulas, from some chosen formulas called \emph{axioms}. 
A \emph{theory} $T$ is a set of theorems 
which is \emph{deductively closed}, in the sense that 
every theorem which can be derived from $T$ using the rules of the logic 
is already in $T$. 
We describe a set-theoretic \emph{intended model} 
for each logic we introduce; 
the rules of the logic are designed so as to 
be \emph{sound} with respect to this intended model. 
Given a logic $L$, the theories of $L$ are partially ordered by inclusion. 
There is a maximal theory $T_\mymax$, where all formulas are theorems. 
There is a minimal theory $T_\mymin$, which is generated by 
the empty set of axioms. 
For all theories $T$ and $T'$, we denote by $T+T'$ the theory 
generated from $T$ and~$T'$. 

\begin{example}
\label{ex:eqn}
With this point of view there are many different 
\emph{equational logics}, with the same deduction rules 
but with different languages, depending on the definition of \emph{terms}.  
In an equational logic, formulas are \emph{pairs of parallel terms} 
$(f,g):X\to Y$
and theorems are \emph{equations} $f\equiv g:X\to Y$. 
Typically, the language of an equational logic 
may be defined from a \emph{signature} (made of sorts and operations).  
The deduction rules are such that the equations in a theory form 
a \emph{congruence}, i.e., an equivalence relation compatible with 
the structure of the terms.
For instance, we may consider the logic ``of naturals'' $L_\nat$, 
with its language generated from the signature made of a sort $N$, 
a constant $0:\unit\to N$ and an operation $s:N\to N$. 
For this logic, the minimal theory is the theory ``of naturals'' $T_\nat$,
the maximal theory is such that $s^k\equiv s^\ell$ 
and $s^k\circ0\equiv s^\ell\circ0$ for all natural numbers $k$ and $\ell$,
and (for instance) the theory ``of naturals modulo~6'' $T_\modsix$
can be generated from the equation $s^6\equiv \id_N$. 
We consider models of equational logics in sets: 
each type $X$ is interpreted as a set (still denoted $X$), 
which is a singleton when $X$ is $\unit$, 
each term $f:X\to Y$ as a function from $X$ to $Y$  
(still denoted $f:X\to Y$),
and each equation as an equality of functions. 
\end{example}

\begin{definition}
\label{defi:hpc}
Given a logic $L$ and its maximal theory $T_\mymax$, 
a theory $T$ is \emph{consistent} if $T\ne T_\mymax$,
and it is \emph{Hilbert-Post complete} 
if it is consistent 
and if any theory containing $T$ 
coincides with~$T_\mymax$ or with~$T$. 
\end{definition}

\begin{example}
\label{ex:hpc-eqn}
In Example~\ref{ex:eqn} we considered two theories for the logic $L_\nat$:  
the theory ``of naturals'' $T_\nat$ 
and the theory ``of naturals modulo~6'' $T_\modsix$.
Since both are consistent and $T_\modsix$ contains $T_\nat$, 
the theory $T_\nat$ is not Hilbert-Post complete.
A Hilbert-Post complete theory for $L_\nat$ is made of 
all equations but $s\equiv \id_N$, it can be generated from the axioms 
$s\!\circ\! 0 \!\equiv\! 0$ and $s\!\circ\! s \!\equiv\! s$.
\end{example}

If a logic $L$ is an extension of a sublogic $L_0$, 
each theory $T_0$ of $L_0$ generates a theory $F(T_0)$ of $L$.
Conversely, each theory $T$ of $L$ determines a theory $G(T)$ of $L_0$, 
made of the theorems of $T$ which are formulas of $L_0$, 
so that $G(T_\mymax)=T_\mymaxz$. 
The functions $F$ and $G$ are monotone and 
they form a \emph{Galois connection}, denoted $F\dashv G$: 
for each theory $T$ of $L$ and each theory $T_0$ of $L_0$ we have 
$F(T_0)\subseteq T$ if and only if $T_0\subseteq G(T)$.
It follows that $T_0\subseteq G(F(T_0))$ and $F(G(T))\subseteq T$. 
Until the end of Section~\ref{sec:hpc}, we consider: 
\emph{a logic $L_0$, an extension $L$ of $L_0$, 
and the associated Galois connection $F\dashv G$.}
%
\begin{definition} 
\label{defi:hpc-rel}
A theory $T'$ of~$L$ is \emph{$L_0$-derivable} 
from a theory $T$ of~$L$ if $T'=T+F(T'_0)$ 
for some theory $T'_0$ of~$L_0$.
A theory $T$ of~$L$ is \emph{(relatively) Hilbert-Post complete with respect to} $L_0$
if it is consistent 
and if any theory of~$L$ containing $T$ 
is $L_0$-derivable from~$T$. 
\end{definition}

Each theory $T$ is $L_0$-derivable from itself, 
as $T=T+F(T_\myminz)$, where $T_\myminz$ 
is the minimal theory of $L_0$. 
In addition, Theorem~\ref{theo:hpc} shows that 
relative completeness 
lifts the usual ``absolute'' completeness from $L_0$ to $L$,
and Proposition~\ref{prop:hpc-compose} proves that 
relative completeness is well-suited to the combination of effects. 
\begin{lemma}
\label{lemm:hpc-rel}
For each theory $T$ of $L$,   
a theory $T'$ of $L$ is $L_0$-derivable from $T$ 
if and only if $T'=T+F(G(T'))$. 
As a special case, $T_\mymax$ is $L_0$-derivable from $T$ 
if and only if $T_\mymax=T+F(T_\mymaxz)$. 
A theory $T$ of $L$ is Hilbert-Post complete with respect to $L_0$
if and only if it is consistent 
and every theory $T'$ of~$L$ containing $T$ 
is such that $T'=T+F(G(T'))$. 
\end{lemma}
\begin{proof}
Clearly, if $T'=T+F(G(T'))$ then $T'$ is $L_0$-derivable from $T$. 
So, let $T'_0$ be a theory of~$L_0$ such that $T'=T+F(T'_0)$,
and let us prove that $T'=T+F(G(T'))$. 
For each theory $T'$ we know that $F(G(T')) \subseteq T'$;
since here $T \subseteq T'$ we get $T+F(G(T')) \subseteq T'$. 
Conversely, 
for each theory $T'_0$ we know that $T'_0 \subseteq G(F(T'_0))$ 
and that $G(F(T'_0)) \subseteq  G(T)+ G(F(T'_0)) \subseteq  G(T+F(T'_0)) $, 
so that $T'_0 \subseteq G(T+F(T'_0)) $;  
since here $T'=T+F(T'_0)$ we get first $T'_0 \subseteq G(T')$ 
and then $T'\subseteq T+F(G(T')) $.
Then, the result for $T_\mymax$ comes from the fact that $G(T_\mymax)=T_\mymaxz $.
The last point follows immediately.  
\end{proof}

\begin{theorem} 
\label{theo:hpc}
Let $T_0$ be a theory of $L_0$ and $T=F(T_0)$. 
If $T_0$ is Hilbert-Post complete (in $L_0$) 
and $T$ is Hilbert-Post complete with respect to $L_0$, 
then $T$ is Hilbert-Post complete (in $L$). 
\end{theorem} 

\begin{proof} 
Since $T$ is complete with respect to $L_0$, it is consistent. 
Since $T=F(T_0)$ we have $T_0 \subseteq G(T)$. 
Let $T'$ be a theory such that $T\subseteq T'$.  
Since $T$ is complete with respect to $L_0$, 
by Lemma~\ref{lemm:hpc-rel} we have $T'=T+F(T'_0)$ where $T'_0=G(T')$.
Since $T\subseteq T'$, $T_0 \subseteq G(T)$ and $T'_0=G(T')$, 
we get $T_0 \subseteq T'_0$. 
Thus, since $T_0$ is complete, either $T'_0=T_0$ or $T'_0=T_\mymaxz$;
let us check that then either $T'=T$ or $T'=T_\mymax$. 
If $T'_0=T_0$ then $F(T'_0)=F(T_0)=T$, so that $T'=T+F(T'_0)=T$. 
If $T'_0=T_\mymaxz$ then $F(T'_0)=F(T_\mymaxz)$; 
since $T$ is complete with respect to $L_0$, 
the theory $T_\mymax$ is $L_0$-derivable from $T$, 
which implies (by Lemma~\ref{lemm:hpc-rel}) that 
$T_\mymax=T+F(T_\mymaxz)=T'$. 
\end{proof} 

\begin{proposition} 
\label{prop:hpc-compose}
Let $L_1$ be an intermediate logic between $L_0$ and $L$, 
let $F_1\dashv G_1$ and $F_2\dashv G_2$ be the Galois connections
associated to the extensions $L_1$ of $L_0$ and $L$ of $L_1$, 
respectively. 
Let $T_1=F_1(T_0)$. 
If $T_1$ is Hilbert-Post complete with respect to $L_0$ 
and $T$ is Hilbert-Post complete with respect to $L_1$ 
then $T$ is Hilbert-Post complete with respect to $L_0$. 
\end{proposition} 

\begin{proof} 
This is an easy consequence of the fact that $F=F_2\circ F_1$.
\end{proof}

Corollary~\ref{coro:hpc-equations} provides a characterization 
of relative Hilbert-Post completeness which is used 
in the next Sections and in the Coq implementation. 

\begin{definition} 
\label{defi:hpc-equations} 
For each set $E$ of formulas let $\Th(E)$ be the theory generated by $E$; 
and when $E=\{e\}$ let $\Th(e)=\Th(\{e\})$. 
Then two sets $E_1$, $E_2$ of formulas are \emph{$T$-equivalent} 
if $T+\Th(E_1)=T+\Th(E_2)$;
and a formula $e$ of~$L$ is \emph{$L_0$-derivable} 
from a theory $T$ of~$L$ if $\{e\}$ is $T$-equivalent to $E_0$ 
for some set $E_0$ of formulas of~$L_0$. 
\end{definition}

\begin{proposition}
\label{prop:hpc-equations}
Let $T$ be a theory of $L$.  
Each theory $T'$ of $L$ containing $T$ is $L_0$-derivable from $T$ 
if and only if each formula $e$ in $L$ is $L_0$-derivable from $T$. 
\end{proposition} 

\begin{proof}    
Let us assume that each theory $T'$ of $L$ containing $T$ 
is $L_0$-derivable from $T$.
Let $e$ be a formula in $L$, let $T'=T+\Th(e)$, 
and let $T'_0$ be a theory of $L_0$ such that $T'=T+F(T'_0)$. 
The definition of $\Th(-)$ is such that $\Th(T'_0)=F(T'_0)$, 
so that we get $T+\Th(e)=T+\Th(E_0)$ where $E_0= T'_0$.
Conversely, 
let us assume that each formula $e$ in $L$ is $L_0$-derivable from $T$.
Let $T'$ be a theory containing $T$.
Let $T''=T+F(G(T'))$, so that $T\subseteq T''\subseteq T'$ 
(because $F(G(T'))\subseteq T'$ for any $T'$).  
Let us consider an arbitrary formula $e$ in $T'$, 
by assumption there is a set $E_0$ of formulas of $L_0$ such that 
$T+\Th(e)=T+\Th(E_0)$. 
Since $e$ is in $T'$ and $T\subseteq T'$ we have $T+\Th(e)\subseteq T'$,
so that $T+\Th(E_0)\subseteq T'$. 
It follows that $E_0$ is a set of theorems of $T'$ which are formulas of $L_0$,
which means that $E_0\subseteq G(T')$, 
and consequently $\Th(E_0)\subseteq F(G(T'))$, 
so that $T+\Th(E_0)\subseteq T''$.
Since $T+\Th(e)=T+\Th(E_0)$ we get $e\in T''$.
We have proved that $T'=T''$, so that $T'$ is $L_0$-derivable from~$T$. 
\end{proof}

\begin{corollary}
\label{coro:hpc-equations}
A theory $T$ of $L$ is Hilbert-Post complete with respect to $L_0$
if and only if it is consistent and for each formula $e$ of $L$ 
there is a set $E_0$ of formulas of~$L_0$ such that 
$\{e\}$ is $T$-equivalent to $E_0$.
\end{corollary}

\section{Completeness for exceptions}
\label{sec:exc}

Exception handling is provided by most modern programming
languages. It allows to deal with anomalous or exceptional events
which require special processing. E.g., one can easily
and simultaneously compute dynamic evaluation in exact linear algebra
using exceptions \cite{DDER14-exc}. 
There, we proposed to deal with exceptions as a decorated effect: 
a term $f:X\to Y$ is not interpreted as a function $f: X\to Y$ unless it is
pure. A term which may raise an exception is instead 
interpreted as a function $f: X\to Y + E$ 
where ``+'' is the disjoint union operator and $E$ is the set of exceptions. 
In this section, we prove the relative Hilbert-Post completeness 
of the decorated theory of exceptions in Theorem~\ref{theo:exc-complete}.

As in \cite{DDER14-exc}, decorated logics for exceptions 
are obtained from equational logics by classifying terms. 
Terms are classified as \emph{pure} terms or \emph{propagators}, 
which is expressed by adding a \emph{decoration} 
or superscript, respectively $\pure$ or $\ppg$;
decoration and type information about terms 
may be omitted when they are clear from the context
or when they do not matter. 
All terms must propagate exceptions, 
and propagators are allowed to raise an exception while pure terms are not. 
The fact of catching exceptions is hidden: 
it is embedded into the $\trycatch$ construction, as explained below. 
In Section~\ref{sec:excore}
we consider a translation 
of the $\trycatch$ construction in a more elementary language 
where some terms are \emph{catchers}, 
which means that they may recover from an exception, 
i.e., they do not have to propagate exceptions.

Let us describe informally a decorated theory for exceptions 
and its intended model. 
Each type $X$ is interpreted as a set, still denoted $X$.
The intended model is described with respect to a set $E$   
called the \emph{set of exceptions}, which does not appear in the syntax.  
A pure term $u^\pure:X\to Y$ is interpreted as a function $u:X\to Y$ 
and a propagator $a^\ppg:X\to Y$ as a function $a:X\to Y+E$;
equations are interpreted as equalities of functions. 
There is an obvious conversion from pure terms to propagators, 
which allows to consider all terms as propagators whenever needed;
if a propagator $a^\ppg:X\to Y$ ``is'' a pure term, 
in the sense that it has been obtained by conversion from a pure term, 
then the function $a:X\to Y+E$ is such that $a(x)\in Y$ for each $x\in
X$. 
This means that exceptions 
are always propagated: the interpretation of 
$(b\circ a)^\ppg:X\to Z$ where $a^\ppg:X\to Y$ and $b^\ppg:Y\to Z$ 
is such that $(b\circ a)(x)=b(a(x))$ when $a(x)$ is not an exception
and $(b\circ a)(x)=e$ when $a(x)$ is the exception $e$
(more precisely, the composition of propagators is the Kleisli composition
associated to the monad $X+E$ \cite[\S~1]{Moggi91}).
Then, exceptions may be classified according to their \emph{name},  
as in \cite{DDER14-exc}. 
Here, in order to focus on the main features of the 
proof of completeness, we assume that there is only one exception name. 
Each exception is built by \emph{encapsulating} a parameter.
Let $P$ denote the type of parameters for exceptions. 
The fundamental operations for raising exceptions are the propagators 
$\throw_Y^\ppg :P\to Y$ for each type $Y$: 
this operation throws an exception with a parameter $p$ of type $P$ 
and pretends that this exception has type $Y$.
The interpretation of the term $\throw_Y^\ppg :P\to Y$ 
is a function $\throw_Y :P\to Y+E$ such that $\throw_Y(p) \in E$
for each $p\in P$. 
The fundamental operations for handling exceptions are the propagators 
$(\try (a) \catch (b))^\ppg:X\to Y$ 
for each terms $a:X\to Y$ and $b:P\to Y$:
this operation first runs $a$ until an exception with parameter $p$ 
is raised (if any), then, if such an exception has been raised, 
it runs $b(p)$.  
The interpretation of the term $(\try (a) \catch (b))^\ppg:X\to Y$ 
is a function $\try (a) \catch (b) :X\to Y+E$ such that 
$(\try (a) \catch (b))(x)=a(x)$ when $a$ is pure and 
$(\try (a) \catch (b))(x)=b(p)$ when $a(x)$ throws an exception 
with parameter $p$.

More precisely, first the definition of the 
\emph{monadic equational logic} $L_\eqn$ is recalled in Fig.~\ref{fig:eqn}, 
(as in~\cite{Moggi91}, this terminology might be misleading: 
the logic is called \emph{monadic} because all its operations are have 
exactly one argument, this is unrelated to the use of the \emph{monad} 
of exceptions).  
\begin{figure}[ht]
\vspace{-5pt}
\begin{framed}
\small
\renewcommand{\arraystretch}{1.5}
\begin{tabular}{l}
Terms are closed under composition: \\  
$u_k\circ \dots \circ u_1:X_0\sto X_k$ 
for each $(u_i:X_{i-1}\sto X_i)_{1\leq i\leq k}$,
and $\id_X:X\sto X$ when $k=0$ \\ 
Rules: 
\quad (equiv) \squad
  $\dfrac{u}{u \eqs u} \quad
  \dfrac{u \eqs v}{v \eqs u}  \quad
  \dfrac{u \eqs v \squad v \eqs w}{u \eqs w}$ \\ 
\quad (subs) \; 
  $\dfrac{u\colon X\to Y \squad v_1 \eqs v_2 \colon Y\to Z}
    {v_1 \circ u \eqs v_2\circ u } $ \squad 
(repl) \; 
  $ \dfrac{v_1\eqs v_2\colon X\to Y \squad w\colon Y\to Z}
    {w\circ v_1 \eqs w\circ v_2} $ \\ 
Empty type $\empt$ with terms $\copa_Y:\empt\to Y$ and rule: 
\quad (initial) \; 
  $ \dfrac{u\colon \empt\to Y }
    {u \eqs \copa_Y} $   \\ 
\end{tabular}
\renewcommand{\arraystretch}{1}
\normalsize 
\vspace{-2mm}
\end{framed}
\vspace{-5pt}
\caption{Monadic equational logic $L_\eqn$ (with empty type)} 
\label{fig:eqn} 
\end{figure}

A monadic equational logic is made of types, terms and operations, 
where all operations are unary, so that terms are simply paths. 
This constraint on arity  
will make it easier to focus on the completeness issue.
For the same reason, we also assume that there is an \emph{empty type}
$\empt$, which is defined as an \emph{initial object}:  
for each $Y$ there is a unique term $\copa_Y:\empt\to Y$ 
and each term $u^\pure:Y\to\empt$ is the inverse of $\copa_Y^\pure$. 
In the intended model, $\empt$ is interpreted as the empty set.

Then, the monadic equational logic $L_\eqn$ is extended to form the
\emph{decorated logic for exceptions} $L_\exc$ 
by applying the rules in Fig.~\ref{fig:exc}, 
with the following intended meaning: 
\begin{itemize}
\item (initial$_1$): the term $\![\,]_Y\!\!$ is 
unique as a propagator, not only as a pure term.
\item (propagate): exceptions are always propagated.
\item (recover): the parameter used for throwing an exception
  may be recovered.
\item (try): equations are preserved by the exceptions mechanism. 
\item (try$_0$): pure code inside $\try$ never triggers 
  the code inside $\catch$.
\item (try$_1$): code inside $\catch$ is executed when 
  an exception is thrown inside~$\try$. 
\end{itemize}
\begin{figure}[ht]
\vspace{-5pt}
\begin{framed}
\small
\renewcommand{\arraystretch}{1.5}
\begin{tabular}{l}
Pure part: the logic $L_\eqn$ with a distinguished type $P$ \\ 
Decorated terms: $\throw_Y^\ppg:P\to Y$ for each type $Y$, \\ 
\quad $(\try (a) \catch (b))^\ppg:X\to Y$ for each $a^\ppg:X\to Y$
and $b^\ppg:P\to Y$, and \\ 
\quad $(a_k\circ \dots \circ a_1)^{(\max(d_1,...,d_k))}:X_0\to X_k$ 
for each $(a_i^{(d_i)}:X_{i-1}\to X_i)_{1\leq i\leq k}$ \\  
\quad with conversion from $u^\pure:X\to Y$ to $u^\ppg:X\to Y$ \\ 
Rules:  \\ 
\quad (equiv), (subs), (repl) for all decorations  \qquad  
(initial$_1$) \; 
  $ \dfrac{a^\ppg\colon \empt\to Y }
    {a \eqs \copa_Y} $ \\ 
\quad (recover) \squad 
  $\dfrac{u_1^\pure,u_2^\pure:X\to P \squad
  \throw_Y\circ u_1 \eqs \throw_Y\circ u_2}{u_1 \eqs u_2} $ \\ 
\quad (propagate)
  $\dfrac{a^\ppg:X\to Y}{a\circ \throw_X \eqs\throw_Y} $ 
\quad (try)
  $\dfrac{a_1^\ppg \eqs a_2^\ppg\!:\!X\to Y \squad b^\ppg\!:\!P\to Y}{
  \try (a_1) \catch (b) \eqs \try (a_2) \catch (b) } $ \\
\quad (try$_0$) \squad 
  $\dfrac{u^\pure\!:\!X\to Y \squad b^\ppg\!:\!P\to Y}{
  \try (u) \catch (b) \eqs u} $  \quad 
(try$_1$) \squad 
  $\dfrac{u^\pure\!:\!X\to P \squad b^\ppg\!:\!P\to Y}{
  \try (\throw_Y \!\circ u) \catch (b) \eqs b\circ u} $ \\
\end{tabular}
\renewcommand{\arraystretch}{1}
\normalsize 
\vspace{-2mm}
\end{framed}
\vspace{-5pt}
\caption{Decorated logic for exceptions $L_\exc$} 
\label{fig:exc} 
\end{figure}
The \emph{theory of exceptions} $T_\exc$ is the theory of $L_\exc$ 
generated from some arbitrary consistent theory $T_\eqn$ of $L_\eqn$; 
with the notations of Section~\ref{sec:hpc}, $T_\exc=F(T_\eqn)$.
The soundness of the intended model follows: see 
\cite[\S 5.1]{DDER14-exc} and~\cite{DDFR12-dual},
which are based on the description
of exceptions in Java~\cite[Ch. 14]{java} or in C++~\cite[\S 15]{cpp}. 

\begin{example}\label{ex:trycatch}
Using the naturals for $P$ and the
successor and predecessor functions (resp. denoted $\successor$ and
$\predecessor$) we can prove, e.g., that
$\try(\successor(\throw~3))\catch(\predecessor)$ is equivalent to $2$.
Indeed, first the rule (propagate) shows that 
$\successor(\throw~3))\equiv\throw~3$,
then the rules (try) and (try$_1$) rewrite the given term into 
$\predecessor(3)$.
\end{example}

Now, in order to prove the completeness of the 
decorated theory for exceptions, we follow a classical method (see,
e.g., \cite[Prop 2.37~\&~2.40]{Pretnar10}): 
we first determine canonical forms in Proposition~\ref{prop:exc-canonical}, 
then we study the equations between terms in canonical form 
in Proposition~\ref{prop:exc-equations}.

\begin{proposition} 
\label{prop:exc-canonical} 
For each $a^\ppg\!:\!X\!\to\! Y$, 
either there is a pure term $u^\pure\!:\!X\!\to\! Y$ such that $a\!\eqs\! u$
or there is a pure term $u^\pure\!:\!X\!\to\! P$ such that 
$a\!\eqs\! \throw_Y \!\circ\! u$. 
\end{proposition}

\begin{proof}
The proof proceeds by structural induction. 
If $a$ is pure the result is obvious, otherwise 
$a$ can be written in a unique way as  
$a = b \circ \mathtt{op} \circ v$ where $v$ is pure, $\mathtt{op}$ 
is either $\throw_Z$ for some $Z$ or $\try(c)\catch(d)$ 
for some $c$ and $d$, and $b$ is the remaining part of $a$. 
If $a = b^\ppg \circ \throw_Z \circ v^\pure$, then by (propagate) 
$a \eqs \throw_Y \circ v^\pure$. 
Otherwise, $a = b^\ppg \circ (\try(c^\ppg)\catch(d^\ppg)) \circ v^\pure$, 
then by induction we consider two cases. 
  \begin{itemize}
  \item If $c \eqs w^\pure$ then by (try$_0$) 
  $a \eqs b^\ppg \circ w^\pure \circ v^\pure$ and 
  by induction we consider two subcases:  
    if $b \eqs t^\pure$ then $a \eqs (t \circ w \circ v)^\pure$ 
    and if $b \eqs \throw_Y\circ t^\pure$ then 
    $a \eqs \throw_Y \circ (t \circ w \circ v)^\pure$. 
    \item If $c \eqs \throw_Z \circ w^\pure$ then by (try$_1$) 
  $a \eqs b^\ppg \circ d^\ppg \circ w^\pure \circ v^\pure$ and 
  by induction we consider two subcases: 
    if $b \circ d \eqs t^\pure$ then $a \eqs (t \circ w \circ v)^\pure$ 
    and if $b \circ d \eqs \throw_Y \circ t^\pure$ then 
    $a \eqs \throw_Y\circ (t \circ w \circ v)^\pure$. 
  \end{itemize}
\end{proof}

Thanks to Proposition~\ref{prop:exc-canonical}, 
the study of  equations in the logic $L_\exc$
can be restricted to pure terms 
and to propagators of the form $\throw_Y \circ v$ where $v$ is pure. 

\begin{proposition}
\label{prop:exc-equations} 
For all $v_1^\pure,v_2^\pure:X\to P$ 
let $a_1^\ppg = \throw_Y \circ v_1 :X\to Y$
and $a_2^\ppg = \throw_Y \circ v_2 :X\to Y$.
Then 
$ a_1^\ppg \eqs a_2^\ppg $ is $T_\exc$-equivalent to 
$v_1^\pure \eqs v_2^\pure $.
\end{proposition}

\begin{proof}
Clearly, if $v_1\eqs v_2$ then $a_1\eqs a_2 $. 
Conversely, if $a_1\eqs a_2 $, i.e., if 
$\throw_Y \circ v_1\eqs \throw_Y \circ v_2$, 
then by rule (recover) it follows that $v_1 \eqs v_2$. 
\end{proof}

In the intended model, for all $v_1^\pure:X\to P$ and $v_2^\pure: X\to Y$, 
it is impossible to have $\throw_Y(v_1(x))=v_2(x)$ for some $x\in X$, 
because $\throw_Y(v_1(x))$ is in the $E$ summand 
and $v_2(x)$ in the $Y$ summand of the disjoint union $Y+E$. 
This means that the functions $\throw_Y \circ v_1$ and $v_2$  
are distinct, as soon as their domain $X$ is a non-empty set. 
For this reason, it is sound to make the following 
Assumption~\ref{ass:exc-equations}. 

\begin{assumption}
\label{ass:exc-equations} 
In the logic $L_\exc$, 
the type of parameters $P$ is non-empty, and  
for all $v_1^\pure:X\to P$ and $v_2^\pure: X\to Y$ with $X$ non-empty, 
let $a_1^\ppg = \throw_Y \circ v_1 :X\to Y$. 
Then 
$ a_1^\ppg \eqs v_2^\pure$ is $T_\exc$-equivalent to $T_\mymaxz$. 
\end{assumption}

\begin{theorem}
\label{theo:exc-complete}
Under Assumption~\ref{ass:exc-equations},
the theory of exceptions $T_\exc$ is Hilbert-Post complete
with respect to the pure sublogic $L_\eqn$ of $L_\exc$. 
\end{theorem}

\begin{proof} 
Using Corollary~\ref{coro:hpc-equations},
the proof relies upon Propositions~\ref{prop:exc-canonical} and  
\ref{prop:exc-equations}.
The theory $T_\exc$ is consistent, because (by soundness)  
it cannot be proved that $\throw_P^\ppg \eqs \id_P^\pure$. 
Now, let us consider an equation between terms with domain $X$
and let us prove that it is $T_\exc$-equivalent to a set of pure equations.  
When $X$ is non-empty, 
Propositions~\ref{prop:exc-canonical} and~\ref{prop:exc-equations}, 
together with Assumption~\ref{ass:exc-equations}, 
prove that the given equation is $T_\exc$-equivalent to a 
set of pure equations. 
When $X$ is empty, then all terms from $X$ to $Y$ are equivalent to $\copa_Y$
so that the given equation 
is $T_\exc$-equivalent to the empty set of pure equations. 
\end{proof}

\section{Completeness of the core language for exceptions}
\label{sec:excore}

In this section, following~\cite{DDER14-exc},  
we describe a translation of the language for exceptions 
from Section~\ref{sec:exc} in a \emph{core} language with
\emph{catchers}.
Thereafter, in Theorem~\ref{theo:excore-complete}, we state the
relative Hilbert-Post completeness of this core language.
Let us call the usual language for exceptions with $\throw$ and $\trycatch$, 
as described in Section~\ref{sec:exc},
the \emph{programmers' language} for exceptions. 
The documentation on the behaviour of exceptions in many languages 
(for instance in Java \cite{java}) 
makes use of a \emph{core language} for exceptions 
which is studied in \cite{DDER14-exc}.
In this language, the empty type plays an important role and 
the fundamental operations for dealing with exceptions 
are $\tagg^\ppg:P\to \empt$ for encapsulating a parameter inside an exception  
and $\untag^\ctc:\empt\to P$ for recovering its parameter 
from any given exception. The new decoration $\ctc$ corresponds to 
\emph{catchers}: a catcher may recover from an exception,
it does not have to propagate it.  
Moreover, the equations also are decorated: 
in addition to the equations '$\eqs$' as in Section~\ref{sec:exc},  
now called \emph{strong equations}, there are 
\emph{weak equations} denoted '$\eqw$'. 

As in Section~\ref{sec:exc}, a set $E$ of exceptions is chosen;  
the interpretation is extended as follows: 
each catcher $f^\ctc:X\to Y$ is interpreted as a function $f:X+E\to Y+E$, 
and there is an obvious conversion from propagators to catchers; 
the interpretation of the composition of catchers is straightforward,
and it is compatible with the Kleisli composition for propagators.  
Weak and strong equations coincide on propagators,
where they are interpreted as equalities,  
but they differ on catchers: 
$f^\ctc \eqw g^\ctc:X\to Y$ means that the functions 
$f,g:X+E\to Y+E$ coincide on $X$, but maybe not on $E$. 
The interpretation of $\tagg^\ppg:P\to \empt$
is an injective function $\tagg:P\to E$ and the interpretation of 
$\untag^\ctc:\empt\to P$ is a function $\untag:E\to P+E$ such that 
$\untag(\tagg(p))=p$ for each parameter $p$.
Thus, the fundamental axiom relating $\tagg^\ppg$ and $\untag^\ctc$ 
is the weak equation $ \untag \circ \tagg \eqw \id_P$.  

\begin{figure}[ht]
\vspace{-5pt}
\begin{framed}
\small
\renewcommand{\arraystretch}{1.5}
\begin{tabular}{l} 
Pure part: the logic $L_\eqn$ with a distinguished type $P$ \\ 
Decorated terms: $\tagg^\acc \colon P\to \empt$, 
$\untag^\ctc \colon \empt\to P$, and \\ 
\quad $(f_k\circ \dots \circ f_1)^{(\max(d_1,...,d_k))}:X_0\to X_k$ 
for each $(f_i^{(d_i)}:X_{i-1}\to X_i)_{1\leq i\leq k}$ \\  
\quad with conversions from $f^\pure$ to $f^\ppg$ 
and from $f^\ppg$ to $f^\ctc$ \\ 
Rules:  \\ 
\quad (equiv$_\eqs$), (subs$_\eqs$), (repl$_\eqs$) for all decorations  \\ 
\quad (equiv$_\eqw$), (repl$_\eqw$) for all decorations, 
  (subs$_\eqw$) only when $h$ is pure \\ 
\quad (empty$_\eqw$) \squad 
  $\dfrac{f\colon \empt\to Y}{f \eqw \copa_Y}$ 
\quad ($\eqs$-to-$\eqw$) \squad 
  $\dfrac{f\eqs g}{f\eqw g} $  
\quad (ax) \squad 
  $\dfrac{}{\untag \circ \tagg \eqw \id_P} $  \\  
\quad (eq$_1$) \squad 
  $\dfrac{f_1^{(d_1)}\eqw f_2^{(d_2)}}{f_1\eqs f_2}$ 
  only when $d_1\leq 1$ and $d_2\leq 1$ \\
\quad (eq$_2$) \squad 
  $\dfrac{f_1,f_2\colon X\to Y \;\;
  f_1\eqw f_2 \;\; f_1 \circ \copa_X \eqs f_2 \circ \copa_X}{f_1\eqs f_2}$ \\ 
\quad (eq$_3$) \squad
  $\dfrac{f_1,f_2\colon \empt\to X \quad f_1 \circ \tagg \eqw f_2 \circ \tagg}
    {f_1\eqs f_2}$ \\ 
\end{tabular}
\renewcommand{\arraystretch}{1}
\normalsize
\vspace{-2mm}
\end{framed}
\vspace{-5pt}
\caption{Decorated logic for the core language for exceptions $L_\excore$} 
\label{fig:excore} 
\end{figure}

More precisely, the 
\emph{decorated logic for the core language for exceptions}
$L_\excore$ 
is defined in Fig.~\ref{fig:excore} as an extension of the 
monadic equational logic $L_\eqn$. 
There is an obvious conversion from strong to weak equations ($\eqs$-to-$\eqw$),
and in addition strong and weak equations coincide on propagators 
by rule (eq$_1$). 
Two catchers $f_1^\ctc,f_2^\ctc:X\to Y$ behave in the same way on exceptions 
if and only if $f_1\circ\copa_X \eqs f_2\circ\copa_X :\empt\to Y$, where
$\copa_X:\empt\to X$ builds a term of type $X$ from any exception. 
Then rule (eq$_2$) expresses the fact that weak and strong equations 
are related by the property that $f_1\eqs f_2$ if and only if 
$f_1\eqw f_2$ and $f_1\circ\copa_X \eqs f_2\circ\copa_X$.
This can also be expressed as a pair of weak equations: 
$f_1\eqs f_2$ if and only if 
$f_1\eqw f_2$ and $f_1\circ\copa_X\circ\tagg \eqw f_2\circ\copa_X\circ\tagg$
by rule (eq$_3$).
The \emph{core theory of exceptions} $T_\excore$ is the theory of $L_\excore$ 
generated from the theory $T_\eqn$ of $L_\eqn$. 
Some easily derived properties are stated in Lemma~\ref{lemm:excore-ul};
which will be used repeatedly.
\begin{lemma} 
\label{lemm:excore-ul}
\begin{enumerate}
\item \label{pt:excore-ul-lulu}
For all pure terms $u_1^\pure,u_2^\pure:X\to P$, the equation 
  $ u_1 \eqs u_2 $ is $T_\excore$-equivalent to 
  $\tagg \circ u_1 \eqs \tagg \circ u_2 $ and also to 
  $\untag \circ \tagg \circ u_1 \eqs \untag \circ \tagg \circ u_2 $.
\item \label{pt:excore-ul-l}
For all pure terms $u^\pure:X\to P$, $v^\pure:X\to\empt$, the equation 
  $ u \eqs \copa_P \circ v$ is $T_\excore$-equivalent to 
  $\tagg \circ u \eqs v $.
\end{enumerate}
\end{lemma} 

\begin{proof}
\begin{enumerate}
\item %
Implications from left to right are clear. Conversely, if 
$\untag \circ \tagg \circ u_1 \eqs \untag \circ \tagg \circ u_2 $,  
then using the axiom (ax) and the rule (subs$_\eqw$) we get $u_1 \eqw u_2$. 
Since $u_1$ and $u_2$ are pure this means that $u_1 \eqs u_2$. 
\item %
First, since $\tagg \circ \copa_P :\empt\to\empt $ is a propagator  
we have $\tagg \circ \copa_P \eqs \id_\empt$. 
Now, if $ u \eqs \copa_P \circ v$ then 
$ \tagg \circ u \eqs \tagg \circ \copa_P \circ v \eqs v$.
Conversely, if $\tagg \circ u \eqs v$ then 
$\tagg \circ u \eqs \tagg \circ \copa_P \circ v$, 
and by Point~\ref{pt:excore-ul-lulu} this means that $u \eqs \copa_P \circ v$.
\end{enumerate}
\end{proof}

The operation $\untag$ in the core language can be used for decomposing 
the $\trycatch$ construction in the programmer's language 
in two steps: a step for catching the exception, 
which is nested into a second step inside the $\trycatch$ block: 
this corresponds to a translation 
of the programmer's language in the core language, 
as in \cite{DDER14-exc}, which is reminded below;
then Proposition~\ref{prop:excore-impl} proves the correctness  
of this translation. 
In view of this translation we extend the core language with:
\begin{itemize}
\item for each $b^\ppg:P\to Y$, a catcher $(\catchcore(b))^\ctc:Y\to Y$ 
such that $\catchcore(b) \eqw \id_Y $ and 
$\catchcore(b)\circ\copa_Y \eqs b\circ\untag$:  
if the argument of $\catchcore(b)$ is non-exceptional then nothing is done, 
otherwise the parameter $p$ of the exception is recovered and $b(p)$ is ran.
\item for each $a^\ppg\!:\!X\to Y$ and $k^\ctc\!:\!Y\to Y$, 
a propagator $ (\trycore(a,k))^\ppg :X\to Y$ 
such that $ \trycore(a,k) \eqw k\circ a$: 
thus $\trycore(a,k)$ behaves as $k\circ a$ on non-exceptional arguments,
but it does always propagate exceptions.
\end{itemize}
Then, a translation of the programmer's language of exceptions 
in the core language is easily obtained: 
for each type $Y$, $\throw_Y^\ppg \!=\! \copa_Y \circ \tagg :P\to Y$. 
and for each $a^\ppg\!:\!X\!\to\! Y$, $b^\ppg\!:\!P\!\to\! Y$, 
$(\try (a) \catch (b))^\ppg \!=\! \trycore(a,\catchcore(b))\!:\!X\!\to\! Y$.
This translation is correct: see Proposition~\ref{prop:excore-impl}.

\begin{proposition}
\label{prop:excore-impl}
If the pure term $\copa_Y:\empt\to Y$ is a monomorphism with respect 
to propagators for each type $Y$, 
the above translation of the programmers' language for exceptions 
in the core language is correct. 
\end{proposition}

\begin{proof}
We have to prove that the image of each rule of $L_\exc$ is satisfied. 
It should be reminded that strong and weak equations coincide on $L_\exc$. 
\begin{itemize}
\item (propagate) For each $a^\ppg:X\to Y$,  
the rules of $L_\excore$ imply that $a\circ \copa_X \eqs \copa_Y$, so that 
$a\circ \copa_X \circ \tagg \eqs \copa_Y \circ \tagg $.
\item (recover) For each $u_1^\pure,u_2^\pure:X\to P$, 
if $\copa_Y \circ \tagg \circ u_1 \eqs \copa_Y \circ \tagg \circ u_2 $ 
since $\copa_Y$ is a monomorphism with respect to propagators 
we have $\tagg \circ u_1 \eqs \tagg \circ u_2 $, 
so that, 
by Point~\ref{pt:excore-ul-lulu} in Lemma~\ref{lemm:excore-ul}, 
we get $u_1 \eqs u_2 $. 
\item (try) 
Since $\try(a_i)\catch(b) \eqw \catch(b) \circ a_i$ 
for $i\in\{1,2\}$, we get $\try(a_1)\catch(b) \eqw \try(a_2)\catch(b)$ 
as soon as $a_1\eqs a_2$. 
\item (try$_0$) For each $u^\pure:X\to Y$ and $b^\ppg:P\to Y$, 
we have $\trycore(u,\catchcore(b)) \eqw \catchcore(b) \circ u$ 
and $\catchcore(b) \circ u \eqw u $ 
(because $\catchcore(b)\eqw\id $ and $u$ is pure), 
so that $\trycore(u,\catchcore(b)) \eqw u $. 
\item (try$_1$) For each $u^\pure:X\to P$ and $b^\ppg:P\to Y$, 
we have $\trycore(\copa_Y \circ \tagg \circ u,\catchcore(b)) \eqw 
\catchcore(b) \circ \copa_Y \circ \tagg \circ u$  
and $ \catchcore(b) \circ \copa_Y \eqs b \circ \untag$
so that $\trycore(\copa_Y \circ \tagg \circ u,\catchcore(b)) \eqw 
b \circ \untag \circ \tagg \circ u$. 
We have also $ \untag \circ \tagg \circ u \eqw u$ 
(because $\untag \circ \tagg \eqw\id $ and $u$ is pure), 
so that $\trycore(\copa_Y \circ \tagg \circ u,\catchcore(b)) \eqw 
b \circ u$. 
\end{itemize}
\end{proof}

\begin{example}[Continuation of
  Example~\ref{ex:trycatch}]\label{ex:taguntag} 
We here show that it is possible to separate the matching between
normal or exceptional 
behavior from the recovery after an
exceptional behavior: to prove that
$\try(\successor(\throw~3))\catch(\predecessor)$ is equivalent to $2$
in the core language, we first use the translation to get:
$\trycore(\successor\circ\copa\circ\tagg\circ{}3,\catchcore(\predecessor))$. 
Then
(empty$_\eqw$) shows that
$\successor\circ\copa\tagg\circ{}3\eqw\copa\circ\tagg\circ{}3$.
Now,
the $\trycore$ and $\catchcore$ translations show that 
$\trycore(\copa\circ\tagg\circ{}3,\catchcore(\predecessor)) \eqw 
\catchcore(\predecessor)\circ\copa\circ\tagg\circ{}3 \eqw 
\predecessor\circ\untag\circ\tagg\circ{}3$.
Finally the axiom (ax) and (eq$_1$) give 
$\predecessor\circ{}3\eqs{}2$.
\end{example}

In order to prove the completeness of  
the core decorated theory for exceptions, 
as for the proof of Theorem~\ref{theo:exc-complete}, 
we first determine canonical forms 
in Proposition~\ref{prop:excore-canonical},
then we study the equations between terms in canonical form  
in Proposition~\ref{prop:excore-equations}. 
Let us begin by proving the 
\emph{fundamental strong equation for exceptions}~(\ref{eq:excore-fundamental}):
by replacement in the axiom (ax) 
we get $\tagg \circ \untag \circ \tagg \eqw \tagg $,  
then by rule (eq$_3$):
\begin{equation}
\label{eq:excore-fundamental}
 \tagg \circ \untag \eqs \id_\empt 
\end{equation}

\begin{proposition}
\label{prop:excore-canonical} 
\begin{enumerate}
\item \label{pt:excore-canonical-acc} 
For each propagator $a^\ppg:X\to Y$, either $a$ is pure or 
there is a pure term $v^\pure:X\to P$ such that 
$ a^\ppg \eqs \copa_Y^\pure \circ \tagg^\ppg \circ v^\pure $.
And for each propagator $a^\ppg:X\to \empt$ (either pure or not), 
there is a pure term $v^\pure:X\to P$ such that 
$ a^\ppg \eqs \tagg^\ppg \circ v^\pure $. 
\item \label{pt:excore-canonical-ctc} 
For each catcher $f^\ctc:X\to Y$, either $f$ is a propagator or 
there is an propagator $a^\acc:P\to Y$ and a pure term $u^\pure:X\to P$ 
such that 
$ f^\ctc \eqs a^\ppg \circ \untag^\ctc \circ \tagg^\ppg \circ u^\pure $.
\end{enumerate}
\end{proposition}

\begin{proof}
\begin{enumerate}
\item %
If the propagator $a^\ppg:X\to Y$ is not pure then it contains 
at least one occurrence of $\tagg^\ppg$. 
Thus, it can be written in a unique way as 
$ a = b \circ \tagg \circ v$ for some propagator $b^\ppg:\empt\to Y$ 
and some pure term $v^\pure:X\to P$. 
Since $b^\ppg:\empt\to Y$ we have $b^\ppg\eqs\copa_Y^\pure$, 
and the first result follows.
When $X=\empt$, it follows that $a^\ppg \eqs \tagg^\ppg \circ v^\pure$.
When $a:X\to \empt$ is pure, one has 
$a \eqs \tagg^\ppg \circ (\copa_P \circ a)^\pure$.
\item %
The proof proceeds by structural induction. 
If $f$ is pure the result is obvious, otherwise 
$f$ can be written in a unique way as 
$f = g \circ \mathtt{op} \circ u$ where $u$ is pure, $\mathtt{op}$ 
is either $\tagg$ or $\untag$ and $g$ is the remaining part of $f$.
By induction, either $g$ is a propagator or 
$g \eqs b\circ \untag \circ \tagg \circ v$ 
for some pure term $v$ and some propagator $b$.
So, there are four cases to consider.
(1) If $\mathtt{op}=\tagg$ and $g$ is a propagator then 
$f$ is a propagator. 
(2) If $\mathtt{op}=\untag$ and $g$ is a propagator then 
by Point~\ref{pt:excore-canonical-acc} 
there is a pure term $w$ such that $u \eqs \tagg \circ w$,
so that $f \eqs g^\ppg \circ \untag \circ \tagg \circ w^\pure$. 
(3) If $\mathtt{op}=\tagg$ and 
$g \eqs b^\ppg\circ \untag \circ \tagg \circ v^\pure$ then 
$f \eqs b\circ \untag \circ \tagg \circ v \circ \tagg \circ u$. 
Since $v: \empt\to P$ is pure we have $\tagg \circ v \eqs \id_\empt$,
so that $f \eqs b^\ppg\circ \untag \circ \tagg \circ u^\pure$. 
(4) If $\mathtt{op}=\untag$ and 
$g \eqs b^\ppg \circ \untag \circ \tagg \circ v^\pure$ then 
$f \eqs b \circ \untag \circ \tagg \circ v \circ \untag \circ u$. 
Since $v$ is pure, by (ax) and (subs$_\eqw$) we have 
$ \untag \circ \tagg \circ v \eqw v $. 
Besides, by (ax) and (repl$_\eqw$) we have 
$ v \circ \untag \circ \tagg \eqw v $ and 
$ \untag \circ \tagg \circ v \circ \untag \circ \tagg \eqw 
\untag \circ \tagg \circ v $. %
Since $\eqw$ is an equivalence relation these three weak equations imply 
$ \untag \circ \tagg \circ v \circ \untag \circ \tagg \eqw 
v \circ \untag \circ \tagg $. 
By rule (eq$_3$) we get 
$ \untag \circ \tagg \circ v \circ \untag \eqs v \circ \untag$, 
and by Point~\ref{pt:excore-canonical-acc} there is 
a pure term $w$ such that $u \eqs \tagg \circ w$, 
so that 
$f \eqs (b \circ v)^\ppg \circ \untag \circ \tagg \circ w^\pure$. 
\end{enumerate}
\end{proof}

Thanks to Proposition~\ref{prop:excore-canonical}, 
in order to study equations in the logic $L_\excore$
we may restrict our study to pure terms, 
propagators of the form $\copa_Y^\pure \circ \tagg^\ppg \circ v^\pure$ 
and catchers of the form 
$a^\ppg \circ \untag^\ctc \circ \tagg^\ppg \circ u^\pure$.

\begin{proposition}  
\label{prop:excore-equations} 
\begin{enumerate}
\item 
\label{prop:excore-equations-ctc-ctc} 
For all $a_1^\ppg,a_2^\ppg:P\to Y$ and $u_1^\pure,u_2^\pure:X\to P$, 
let $f_1^\ctc = a_1\circ \untag\circ \tagg\circ u_1:X\to Y$
and $f_2^\ctc = a_2\circ \untag\circ \tagg\circ u_2:X\to Y$,
then $f_1 \eqw f_2 $ is $T_\excore$-equivalent to 
$ a_1\circ u_1 \eqs a_2\circ u_2 $
and $ f_1\eqs f_2 $ is $T_\excore$-equivalent to 
$ \{ a_1\eqs a_2 \;,\; a_1\circ u_1 \eqs a_2\circ u_2 \}$.
\item 
\label{prop:excore-equations-ctc-ppg}
For all $a_1^\ppg:P\to Y$, $u_1^\pure:X\to P$ and $a_2^\ppg:X\to Y$, 
let $ f_1^\ctc = a_1\circ \untag\circ \tagg\circ u_1:X\to Y$,
then 
$f_1 \eqw a_2$ is $T_\excore$-equivalent to 
$ a_1\circ u_1 \eqs a_2 $ and 
$f_1 \eqs a_2$ is $T_\excore$-equivalent to 
$ \{ a_1\circ u_1 \eqs a_2 \;,\; a_1\eqs \copa_Y \circ \tagg \}$.
\item 
\label{pt:excorec-equations-ppg-ppg} 
Let us assume that $\copa_Y^\pure$ is a monomorphism with respect 
to propagators. 
For all $v_1^\pure,v_2^\pure:X\to P$,
let $a_1^\ppg = \copa_Y\circ \tagg \circ v_1:X\to Y$
and $a_2^\ppg = \copa_Y\circ \tagg \circ v_2:X\to Y$. 
Then 
$ a_1\eqs a_2 $ is $T_\excore$-equivalent to 
$ v_1\eqs v_2 $. 
\end{enumerate}
\end{proposition}

\begin{proof} 
\begin{enumerate}
\item %
Rule (eq$_2$) implies that $f_1\eqs f_2$ if and only if 
$f_1\eqw f2$ and $f_1 \circ \copa_X \eqs f_2 \circ \copa_X$. 
On the one hand, $f_1\eqw f_2$ if and only if $a_1\circ u_1 \eqs a_2\circ u_2$:
indeed, for each $i\in\{1,2\}$, by (ax) and (subs$_\eqw$), since $u_i$ is pure 
we have $f_i \eqw a_i \circ u_i$.
On the other hand, let us prove that 
$f_1 \circ \copa_X \eqs f_2 \circ \copa_X$ if and only if $a_1 \eqs a_2$.
For each $i\in\{1,2\}$, 
the propagator $\tagg \circ u_i \circ \copa_X : \empt\to\empt$ satisfies 
$\tagg \circ u_i \circ \copa_X \eqs \id_\empt$,
so that $f_i \circ \copa_X \eqs a_i \circ \untag $.
Thus, $f_1 \circ \copa_X \eqs f_2 \circ \copa_X$ if and only if 
$a_1 \circ \untag \eqs a_2 \circ \untag$.
Clearly, if $a_1 \eqs a_2$ then $a_1 \circ \untag \eqs a_2 \circ \untag$.
Conversely, if $a_1 \circ \untag \eqs a_2 \circ \untag$
then $a_1 \circ \untag \circ \tagg \eqs a_2 \circ \untag \circ \tagg $, 
so that by (ax) and (repl$_\eqw$) we get $a_1 \eqw a_2$,
which means that $a_1 \eqs a_2$ because $a_1$ and $a_2$ are propagators. 
\item %
Rule (eq$_2$) implies that $f_1\eqs a_2$ if and only if 
$f_1\eqw a_2$ and $f_1 \circ \copa_X \eqs a_2 \circ \copa_X$. 
On the one hand, $f_1\eqw a_2$ if and only if $a_1\circ u_1 \eqs a_2$:
indeed, by (ax) and (subs$_\eqw$), since $u_1$ is pure 
we have $f_1 \eqw a_1 \circ u_1$.
On the other hand, let us prove that 
$f_1 \circ \copa_X \eqs a_2 \circ \copa_X$ if and only if 
$a_1\eqs \copa_Y \circ \tagg$, in two steps.
Since $ a_2 \circ \copa_X : \empt\to Y$ is a propagator, 
we have $ a_2 \circ \copa_X \eqs \copa_Y$. 
Since $f_1 \circ \copa_X = 
a_1\circ \untag\circ \tagg\circ u_1 \circ \copa_X$ 
with $\tagg\circ u_1 \circ \copa_X:\empt\to\empt $ a propagator, 
we have $\tagg\circ u_1 \circ \copa_X \eqs \id_\empt$ 
and thus we get $ f_1 \circ \copa_X \eqs a_1\circ \untag$. 
Thus, $f_1 \circ \copa_X \eqs a_2 \circ \copa_X$ if and only if  
$a_1\circ \untag \eqs \copa_Y$. 
If $a_1\circ \untag \eqs \copa_Y$ then 
$a_1\circ \untag\circ \tagg \eqs \copa_Y\circ \tagg$, 
by (ax) and (repl$_\eqw$) this implies 
$a_1 \eqw \copa_Y\circ \tagg$, which is a strong equality because 
both members are propagators. 
Conversely, if $a_1 \eqs \copa_Y\circ \tagg $ 
then $a_1 \circ \untag \eqs \copa_Y\circ \tagg \circ \untag $, 
by the fundamental equation~(\ref{eq:excore-fundamental})
this implies $a_1 \circ \untag \eqs \copa_Y$. 
Thus, $a_1\circ \untag \eqs \copa_Y$ if and only if 
$a_1 \eqs \copa_Y\circ \tagg $. 
\item %
Clearly, if $ v_1\eqs v_2 $ then 
$ \copa_Y\circ \tagg \circ v_1\eqs \copa_Y\circ \tagg \circ v_2$.
Conversely, if $ \copa_Y\circ \tagg \circ v_1\eqs \copa_Y\circ \tagg \circ v_2$
then since $\copa_Y$ is a monomorphism with respect to propagators 
we get $ \tagg \circ v_1\eqs \tagg \circ v_2$, 
so that $ \untag \circ \tagg \circ v_1\eqs \untag \circ \tagg \circ v_2$.
Now, from (ax) we get $ v_1\eqw v_2 $, which means that 
$ v_1\eqs v_2 $ because $v_1$ and $v_2$ are pure. 
\end{enumerate}
\end{proof}

Assumption~\ref{ass:excore-equations} 
is the image of Assumption~\ref{ass:exc-equations}
by the above translation. 

\begin{assumption}
\label{ass:excore-equations} 
In the logic $L_\excore$, 
the type of parameters $P$ is non-empty, and  
for all $v_1^\pure:X\to P$ and $v_2^\pure: X\to Y$ with $X$ non-empty, 
let $a_1^\ppg = \copa_Y\circ \tagg \circ v_1 :X\to Y$. 
Then 
$ a_1^\ppg \eqs v_2^\pure$ is $T_\exc$-equivalent to $T_\mymaxz$. 
\end{assumption}

\begin{theorem} 
\label{theo:excore-complete}
Under Assumption~\ref{ass:excore-equations},
the theory of exceptions $T_\excore$ is Hilbert-Post complete
with respect to the pure sublogic $L_\eqn$ of $L_\excore$. 
\end{theorem}

\begin{proof}
Using Corollary~\ref{coro:hpc-equations},
the proof is based upon Propositions~\ref{prop:excore-canonical} and
\ref{prop:excore-equations}. 
It follows the same lines as the proof of Theorem~\ref{theo:exc-complete}, 
except when $X$ is empty: because of catchers the proof here 
is slightly more subtle.
First, the theory $T_\excore$ is consistent, because (by soundness) 
it cannot be proved that $\untag^\ctc \eqs \copa_P^\pure$.
Now, let us consider an equation between terms $f_1,f_2:X\to Y$, 
and let us prove that it is $T_\excore$-equivalent to a set of pure equations.
When $X$ is non-empty, 
Propositions~\ref{prop:excore-canonical} and~\ref{prop:excore-equations},
together with Assumption~\ref{ass:excore-equations}, 
prove that the given equation is $T_\excore$-equivalent to a 
set of pure equations.
When $X$ is empty, then $f_1\eqw \copa_Y$ and $f_2\eqw \copa_Y$, 
so that if the equation is weak or if both $f_1$ and $f_2$ are propagators 
then the given equation is
$T_\excore$-equivalent to the empty set of equations between pure terms. 
When $X$ is empty and the equation is $f_1 \eqs f_2$ with 
at least one of $f_1$ and $f_2$ a catcher, then 
by Point~\ref{prop:excore-equations-ctc-ctc}  
or~\ref{prop:excore-equations-ctc-ppg} 
of Proposition~\ref{prop:excore-equations},
the given equation is $T_\excore$-equivalent to a set of 
equations between propagators; 
but we have seen that each equation between propagators
(whether $X$ is empty or not) is 
$T_\excore$-equivalent to a set of equations between pure terms,
so that $f_1\eqs f_2$ is $T_\excore$-equivalent to the union of 
these sets of pure equations.
\end{proof}

\section{Verification of Hilbert-Post Completeness in Coq} 
\label{sec:coq}

All the statements of Sections~\ref{sec:exc} and~\ref{sec:excore}  
have been checked in Coq. 
The proofs can be found in 
\url{http://forge.imag.fr/frs/download.php/680/hp-0.7.tar.gz}, 
as well as an almost dual proof for the completeness of the state.
They share the same framework, defined  in~\cite{DDEP13-coq}: 
\begin{enumerate}
\item the terms of each logic are inductively
defined through the dependent type named $\mathtt{term}$ which builds
a new \texttt{Type} out of two input \texttt{Type}s. 
For instance, $\mathtt{term \ Y \ X}$ is the \texttt{Type} of all terms 
of the form $\mathtt{f\colon X\to Y}$;
\item the decorations are enumerated: 
\texttt{pure} and \texttt{propagator} for both languages,  
and \texttt{catcher} for the core language;
\item decorations are inductively assigned to the terms via the
dependent type called $\mathtt{is}$. The latter builds a proposition
(a \texttt{Prop} instance in Coq) out of a \texttt{term} and a
decoration. Accordingly, \texttt{is pure (id X)} is a \texttt{Prop}
instance;
\item for the core language, 
we state the rules with respect to weak and strong equalities
by defining them in a mutually inductive way.
\end{enumerate}

The completeness proof for the exceptions core language
is 950 SLOC in Coq where it is 460 SLOC in \LaTeX.
Full certification runs in 6.745s on a Intel i7-3630QM @2.40GHz
using the Coq Proof Assistant, v. 8.4pl3.
Below table details the proof lengths and timings for each library. 

\begin{center}
\begin{tabular}{ |l||l|l|l|l|}
 \hline
 \multicolumn{5}{|c|}{Proof lengths \& Benchmarks} \\
 \hline
 package    & source  & length  & length  & execution time \\
   &   &  in Coq  & in \LaTeX  & in Coq \\
  \hline
 exc\_cl-hp  \;  & HPCompleteCoq.v \;  &40 KB&   15 KB&6.745 sec. \\
 exc\_pl-hp &   HPCompleteCoq.v  & 8 KB   &6 KB&1.704 sec. \\
 exc\_trans & Translation.v & 4 KB&  2 KB&1.696 sec. \\
 st-hp & HPCompleteCoq.v & 48 KB&  15 KB &7.183 sec.\\
 \hline
\end{tabular}
\end{center}

The correspondence between the propositions and theorems
in this paper and their proofs in Coq is given in Fig.~\ref{fig:coq-table},
and the dependency chart for the main results in Fig.~\ref{fig:coq-chart}. 
For instance, 
Proposition~\ref{prop:exc-equations} is expressed in Coq as: 
\scriptsize
\begin{verbatim}
forall {X Y} (a1 a2: term X Y) (v1 v2: term (Val e) Y),
    (is pure v1) /\ (is pure v2) /\
    (a1 = ((@throw X e) o v1)) /\ (a2 = ((@throw X e) o v2)) -> ((a1 == a2) <-> (v1 == v2)).
\end{verbatim}
\normalsize

\begin{figure}[ht]

\renewcommand{\arraystretch}{1}
\begin{center}
\begin{tabular}{ |l|l|}
\hline
\multicolumn{2}{|c|}{ hp-0.7/exc$\_$trans/Translation.v} \\
\hline
Proposition~\ref{prop:excore-impl} (propagate) & \texttt{propagate} \\
Proposition~\ref{prop:excore-impl} (recover) & \texttt{recover} \\
Proposition~\ref{prop:excore-impl} (try) & \texttt{try} \\
Proposition~\ref{prop:excore-impl} (try$_0$) & \texttt{try$_0$} \\
Proposition~\ref{prop:excore-impl} (try$_1$) & \texttt{try$_1$} \\
\hline
\end{tabular}  \\
\vspace{15pt}
\begin{tabular}{ |l|l|}
\hline
\multicolumn{2}{|c|}{ hp-0.7/exc$\_$pl-hp/HPCompleteCoq.v} \\
\hline
Proposition~\ref{prop:exc-canonical} \; & \texttt{can$\_$form$\_$th} \\
Proposition~\ref{prop:exc-equations} 
& \texttt{eq$\_$th$\_$1$\_$eq$\_$pu} \\
Assumption~\ref{ass:exc-equations} \; & \texttt{eq$\_$th$\_$pu$\_$abs} \\
Theorem~\ref{theo:exc-complete} & \texttt{HPC$\_$exc$\_$pl} \\
\hline
\end{tabular}\\

\begin{tabular}{ |l|l|}
\multicolumn{2}{c}{ \null } \\
\hline
\multicolumn{2}{|c|}{ hp-0.7/exc$\_$cl-hp/HPCompleteCoq.v} \\
\hline
Proposition~\ref{prop:excore-canonical}~Point~\ref{pt:excore-canonical-acc}
& \texttt{can$\_$form$\_$pr} \\
Proposition~\ref{prop:excore-canonical}~Point~\ref{pt:excore-canonical-ctc}
& \texttt{can$\_$form$\_$ca} \\
Assumption~\ref{ass:excore-equations} & \texttt{eq$\_$pr$\_$pu$\_$abs} \\
Proposition~\ref{prop:excore-equations}~Point~\ref{prop:excore-equations-ctc-ctc} & \texttt{eq$\_$ca$\_$2$\_$eq$\_$pr} \\
Proposition~\ref{prop:excore-equations}~Point~\ref{prop:excore-equations-ctc-ppg} & \texttt{eq$\_$ca$\_$pr$\_$2$\_$eq$\_$pr} \\
Proposition~\ref{prop:excore-equations}~Point~\ref{pt:excorec-equations-ppg-ppg} & \texttt{eq$\_$pr$\_$1$\_$eq$\_$pu} \\
Theorem~\ref{theo:excore-complete} & \texttt{HPC$\_$exc$\_$core} \\
\hline
\end{tabular}
\end{center}
\renewcommand{\arraystretch}{1}
\vspace{-15pt}
\caption{Correspondence between theorems in this paper and
  their Coq counterparts} 
\label{fig:coq-table}
\end{figure}
\begin{figure}[H]
\begin{framed}
\hspace{-5mm} 
\small
$ \xymatrix@C=.2pc@R=-.3pc{
\mathtt{can\_form\_ca} \ar[dr] &\\
& \mathtt{eq\_ca\_1\_or\_2\_eq\_pr} \ar[ddr] \ar[ddddddr]\\
\mathtt{eq\_ca\_pr\_2\_eq\_pr} \ar[ur] &\\
& & \mathtt{eq\_ca\_abs\_or\_2\_eq\_pu} \ar[ddr]\\
\mathtt{can\_form\_pr} \ar[dr]\\
\mathtt{eq\_pr\_1\_eq\_pu} \ar[r] & \mathtt{eq\_pr\_abs\_or\_1\_eq\_pu} \ar[uur] \ar[ddr] & &\mathtt{HPC\_exc}\\
\mathtt{eq\_pr\_pu\_abs} \ar[ur]&\\
& & \mathtt{eq\_ca\_abs\_2\_eq\_pu\_dom\_emp}\ar[uur]\\
& \mathtt{eq\_pr\_dom\_emp}\ar[ur]\\
} $
\normalsize
\vspace{-2mm}
\end{framed}
\vspace{-15pt}
\caption{Dependency chart for the main results} 
\label{fig:coq-chart}
\end{figure}
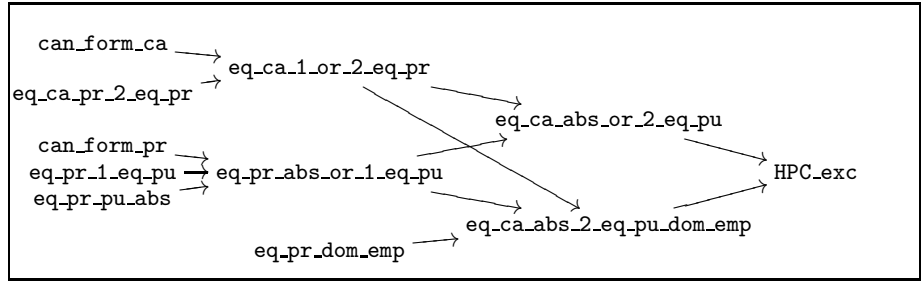

\section{Conclusion and future work}

This paper is a first step towards the proof of completeness 
of decorated logics for computer languages. 
It has to be extended in several directions:  
adding basic features to the language (arity, conditionals, loops, \dots), 
proving completeness of the decorated approach for other effects 
(not only states and exceptions);
the combination of effects should easily follow, 
thanks to Proposition~\ref{prop:hpc-compose}.

\newpage 

\appendix

\section{Completeness for states}
\label{app:sta}

Most programming languages such as C/C++ and Java support the usage
and manipulation of the state (memory) structure. Even though the
state structure is never syntactically mentioned, the commands are
allowed to use or manipulate it, for instance looking up or updating
the value of variables. This provides a great flexibility in
programming, but in order to prove the correctness of programs, one
usually has to revert to an explicit manipulation of the state.
Therefore, any access to the state, regardless of usage or
manipulation, is treated as a computational effect: a
syntactical term $f: X\to Y$ is not interpreted as $f:X\to Y$ 
unless it is {\em pure}, that
is unless it does not use the variables in any manner.  
Indeed, a term which updates the state has instead the following
interpretation: $f: X \times S\to Y \times S$ 
where \lq{}$\times$\rq{} is the product operator and $S$ is the set
of possible states. In~\cite{DDEP13-coq}, we proposed a proof system to
prove program properties involving states effect, while keeping the
memory manipulations implicit. 
We summarize this system next and prove its Hilbert-Post completeness
in Theorem~\ref{theo:sta-complete}.

As noticed in~\cite{DDER14-exc}, the logic $L_\excore$ is exactly 
dual to the logic $L_\sta$ for states (as reminded below).  
Thus, the dual of the completeness Theorem~\ref{theo:excore-complete}
and of all results in Section~\ref{sec:excore} are valid, with the dual proof. 
However, the intended models for exceptions and for states 
rely on the category of sets, which is not self-dual, 
and the additional assumptions in Theorem~\ref{theo:excore-complete},
like the existence of a boolean type,  
cannot be dualized without loosing the soundness of the logic 
with respect to its intended interpretation. 
It follows that the completeness Theorem~\ref{theo:sta-complete}
for the theory for states is not exactly the dual of 
Theorem~\ref{theo:excore-complete}. 
In this Appendix, for the sake of readability, 
we give all the details of the proof of Theorem~\ref{theo:sta-complete}; 
we will mention which parts are \emph{not} the dual of the corresponding 
parts in the proof of Theorem~\ref{theo:excore-complete}.

As in \cite{DDFR12-state}, decorated logics for states 
are obtained from equational logics by classifying terms and equations. 
Terms are classified as \emph{pure} terms, \emph{accessors} 
or \emph{modifiers}, which is expressed by adding a \emph{decoration} 
or superscript, respectively $\pure$, $\acc$ and $\modi$;
decoration and type information about terms 
may be omitted when they are clear from the context
or when they do not matter. 
Equations are classified as \emph{strong} or \emph{weak} equations,
denoted respectively by the symbols $\eqs$ and $\eqw$.
Weak equations relates to the values returned by programs, while
strong equations relates to both values and side effects.
In order to observe the state, accessors may use the values 
stored in \emph{locations}, and modifiers may update these values. 
In order to focus on the main features of the 
proof of completeness, let us assume that 
only one location can be observed and modified; 
the general case, with an arbitrary number of locations, is considered 
in Remark~\ref{rem:sta-complete}.
The logic for dealing with pure terms may be any logic 
which extends a monadic equational logic with constants $L_\eqnunit$; 
its terms are decorated as pure and its equations are strong. 
This \emph{pure sublogic} $L_\sta^\pure$ is extended to form the
corresponding \emph{decorated logic for states} $L_\sta$. 
The rules for $L_\sta$ are given in Fig.~\ref{fig:sta}. 
A theory $T^\pure$ of $L_\sta^\pure$ is chosen,
then the \emph{theory of states} $T_\sta$ is the theory of $L_\sta$ 
generated from $T^\pure$. 
Let us now discuss the logic $L_\sta$ and 
its intended interpretation in sets; 
it is assumed that some model of the pure subtheory $T^\pure$ in sets 
has been chosen; 
the names of the rules refer to Fig.~\ref{fig:sta}. 

Each type $X$ is interpreted as a set, denoted $X$.
The intended model is described with respect to a set $S$ 
called the \emph{set of states}, which does not appear in the syntax.  
A pure term $u^\pure:X\to Y$ 
is interpreted as a function $u:X\to Y$, 
an accessor $a^\acc:X\to Y$
as a function $a:S\times X\to Y$,
and a modifier $f^\modi:X\to Y$
as a function $f:S\times X\to S\times Y$. 
There are obvious conversions from pure terms to accessors
and from accessors to modifiers,
which allow to consider all terms as modifiers whenever needed;
for instance, this allows to interpret the composition of terms 
without mentioning Kleisli composition; the complete characterization
is given in~\cite{DDFR12-state}.  

Here,  for the sake of simplicity, we consider a
single variable (as done, e.g., in~\cite{Pretnar10}
and~\cite{Staton10-fossacs}),
and dually to the choice of a unique exception name in 
Section~\ref{sec:excore}. 
See Remark~\ref{rem:sta-complete} for the 
generalization to an arbitrary number of variables.
The values of the unique location have type $V$.
The fundamental operations for dealing with the state
are the accessor $\lookup^\acc:\unit\to V$ 
for reading the value of the location and the modifier 
$\update^\modi:V\to\unit$ for updating this value.
According to their decorations, they are interpreted respectively 
as functions $\lookup:S\to V$ and $\update:S\times V\to S$. 
Since there is only one location, it might be assumed that 
$\lookup:S\to V$ is a bijection 
and that $\update:S\times V\to S$ maps each $(s,v)\in S\times V$ 
to the unique $s'\in S$ such that $\lookup(s')=v$:
this is expressed by a weak equation, as explained below. 

A strong equation $f\eqs g$ means that $f$ and $g$ 
return the same result and modify the state in ``the same way'', 
which means that no difference can be observed between 
the side-effects performed by $f$ and by $g$.
Whenever $\lookup:S\to V$ is a bijection, 
a strong equation $f^\modi\eqs g^\modi:X\to Y$ 
is interpreted as the equality $f=g:S\times X\to S\times Y$:
for each $(s,x)\in S\times X$, let $f(s,x)=(s',y')$ and 
$g(s,x)=(s'',y'')$, then $f\eqs g$ means that $y'=y''$ and $s'=s''$
for all $(s,x)$.
Strong equations form a congruence. 
A weak equation $f\eqw g$ means that $f$ and $g$ 
return the same result although they may modify the state in different ways.
Thus, a weak equation $f^\modi\eqw g^\modi:X\to Y$ 
is interpreted as the equality 
$\pr_Y\circ f=\pr_Y\circ g:S\times X\to Y$,
where $\pr_Y:S\times Y\to Y$ is the projection; 
with the same notations as above, this means that $y'=y''$ for all $(s,x)$. 
Weak equations do not form a congruence: 
the replacement rule holds only when the replaced term is pure. 
The fundamental equation for states is provided by rule (ax): 
  $\lookup^\acc \circ\update^\modi\eqw \id_V$. 
This means that updating the location with a value $v$ 
and then observing the value of the location does return $v$.
Clearly this is only a weak equation: 
its right-hand side does not modify the state while its 
left-hand side usually does. 
There is an obvious conversion from strong to weak equations ($\eqs$-to-$\eqw$),
and in addition strong and weak equations coincide on accessors 
by rule (eq$_1$). 
Two modifiers $f_1^\modi,f_2^\modi:X\to Y$ modify the state in the same way if 
and only if $\pa_Y\circ f_1 \eqs \pa_Y\circ f_2:X\to \unit$, where
$\pa_Y:Y\to\unit$ throws out the returned value.
Then weak and strong equations are related by the property that 
$f_1\eqs f_2$ if and only if 
$f_1\eqw f_2$ and $\pa_Y\circ f_1 \eqs \pa_Y\circ f_2$, by rule (eq$_2$).
This can be expressed as a pair of weak equations  
$f_1\eqw f_2$ and 
$\lookup\circ \pa_Y\circ f_1 \eqw \lookup\circ \pa_Y\circ f_2$, 
by rule (eq$_3$).
Some easily derived properties are stated in Lemma~\ref{lemm:sta-ul};
Point~\ref{pt:sta-unit} will be used repeatedly.

\begin{figure}[ht]
\begin{framed}
\small
\begin{tabular}{l} 
\textbf{Monadic equational logic with constants $L_\eqnunit$:} \\ 
Types and terms: as for monadic equational logic, plus \\
\quad a unit type $\unit$ and 
a term $\pa_X:X\to \unit$ for each $X$ \\
Rules: as for monadic equational logic, plus 
\quad (unit) \squad 
  $\dfrac{f\colon X\to \unit}{f \eqs \pa_X}$ \\ 
\hline
\textbf{Decorated logic for states $L_\sta$:}  \\ 
Pure part: some logic $L_\sta^\pure$ extending $L_\eqnunit$, 
with a distinguished type $V$ \\ 
Decorated terms: $\lookup^\acc \colon \unit\to V$, 
$\update^\modi \colon V\to \unit$, and \\ 
\quad $(f_k\circ \dots \circ f_1)^{(\max(d_1,...,d_k))}:X_0\to X_k$ 
for each $(f_i^{(d_i)}:X_{i-1}\to X_i)_{1\leq i\leq k}$ \\  
\quad with conversions from $f^\pure$ to $f^\acc$ 
and from $f^\acc$ to $f^\modi$ \\ 
Rules:  \\ 
\quad (equiv$_\eqs$), (subs$_\eqs$), (repl$_\eqs$) for all decorations  \\ 
\quad (equiv$_\eqw$), (subs$_\eqw$) for all decorations, 
  (repl$_\eqw$) only when $h$ is pure \\ 
\quad (unit$_\eqw$) \squad 
  $\dfrac{f\colon X\to \unit}{f \eqw \pa_X}$ 
\quad ($\eqs$-to-$\eqw$) \squad 
  $\dfrac{f\eqs g}{f\eqw g} $  
\quad (ax) \squad 
  $\dfrac{}{\lookup \circ \update \eqw \id_V} $  \\  
\quad (eq$_1$) \squad 
  $\dfrac{f_1^{(d_1)}\eqw f_2^{(d_2)}}{f_1\eqs f_2}$ 
  only when $d_1\leq 1$ and $d_2\leq 1$ \\
\quad (eq$_2$) \squad 
  $\dfrac{f_1,f_2\colon X\to Y \;\;
  f_1\eqw f_2 \;\; \pa_Y\circ f_1 \eqs \pa_Y\circ f_2 }{f_1\eqs f_2}$ \\ 
\quad (eq$_3$) \squad 
  $\dfrac{f_1,f_2\colon X\to \unit \quad \lookup\circ f_1 \eqw \lookup\circ f_2}
    {f_1\eqs f_2}$ \\ 
\end{tabular}
\normalsize
\vspace{-2mm}
\end{framed}
\vspace{-5pt}
\caption{Decorated logic for states (dual to Fig.~\ref{fig:excore}) } 
\label{fig:sta} 
\end{figure}

\begin{lemma} 
\label{lemm:sta-ul}
\begin{enumerate}
\item \label{pt:sta-ul-ul} 
$ \update\circ\lookup \eqs \id_{\unit} $.
(this is the fundamental strong equation for states).
\item \label{pt:sta-unit} 
each $f^\modi \colon \unit\to \unit$ is such that $f \eqw \id_\unit$,
each $f^\acc \colon X\to \unit$ is such that $f \eqs \pa_X $, 
and each $f^\acc \colon \unit\to \unit$  is such that $f \eqs \id_\unit$.
\item \label{pt:sta-ul-lulu}
For all pure terms $u_1^\pure,u_2^\pure:V\to Y$, one has:  
  $ u_1 \eqs u_2  $ is $T_\sta$-equivalent to 
$ u_1 \circ \lookup \eqs u_2 \circ \lookup $ and also to 
$ u_1 \circ \lookup \circ \update \eqs u_2 \circ \lookup\circ \update $.
\item \label{pt:sta-ul-l}
For all pure terms $u^\pure:V\to Y$, $v^\pure:\unit\to Y$, one has:
  $ u \eqs v \circ \pa_V $ is $T_\sta$-equivalent to 
  $ u \circ \lookup \eqs v $.
\end{enumerate}
\end{lemma} 

\begin{proof}
\begin{enumerate}
\item %
By substitution in the axiom (ax) 
we get $\lookup \circ \update \circ \lookup \eqw \lookup $;
then by rule (eq$_3$) $\update\circ\lookup \eqs \id_\unit$.
\item %
Clear.
\item %
Implications from left to right are clear. Conversely, if 
$u_1 \circ \lookup \circ \update \eqs u_2 \circ \lookup\circ \update$,  
then using the axiom (ax) and the rule (repl$_\eqw$) we get $u_1 \eqw u_2$. 
Since $u_1$ and $u_2$ are pure this means that $u_1 \eqs u_2$. 
\item %
First, since $ \pa_V \circ \lookup:\unit\to\unit $ is an accessor  
we have $ \pa_V \circ \lookup \eqs \id_\unit$. 
Now, if $u \eqs v \circ \pa_V$ then 
$u \circ \lookup \eqs v \circ \pa_V \circ \lookup$,
so that $u \circ \lookup \eqs v$.
Conversely, if $u \circ \lookup \eqs v$ then 
$u \circ \lookup \eqs v \circ \pa_V \circ \lookup$, 
and by Point~(\ref{pt:sta-ul-lulu}) this means that $u \eqs v \circ \pa_V$.
\end{enumerate}
\end{proof}

Our main result is Theorem~\ref{theo:sta-complete} about the 
relative Hilbert-Post completeness of the decorated theory of states 
under suitable assumptions. 

\begin{proposition} 
\label{prop:sta-canonical} 
\begin{enumerate}
\item \label{pt:sta-canonical-acc} 
For each accessor $a^\acc:X\to Y$, either $a$ is pure or 
there is a pure term $v^\pure:V\to Y$ such that 
$ a^\acc \eqs v^\pure\circ \lookup^\acc\circ \pa_X^\pure $. 
\\ For each accessor $a^\acc:\unit\to Y$ (either pure or not), 
there is a pure term $v^\pure:V\to Y$ such that 
$ a^\acc \eqs v^\pure\circ \lookup^\acc $. 
\item \label{pt:sta-canonical-modi} 
For each modifier $f^\modi:X\to Y$, either $f$ is an accessor or 
there is an accessor $a^\acc:X\to V$ and a pure term $u^\pure:V\to Y$ 
such that 
$ f^\modi \eqs u^\pure\circ \lookup^\acc\circ \update^\modi \circ a^\acc $.
\end{enumerate}
\end{proposition}

\begin{proof}
\begin{enumerate}
\item %
If the accessor $a^\acc:X\to Y$ is not pure then it contains 
at least one occurrence of $\lookup^\acc$. 
Thus, it can be written in a unique way as 
$ a = v\circ \lookup\circ b$ for some pure term 
$v^\pure:V\to Y$ and some accessor $b^\acc:X\to \unit$.
Since $b^\acc:X\to \unit$ we have $b^\acc\eqs\pa_X^\pure$, 
and the first result follows.
When $X=\unit$, it follows that $a^\acc \eqs v^\pure\circ \lookup^\acc $.
When $a:\unit\to Y$ is pure, one has 
$a \eqs (a\circ\pa_V)^\pure\circ \lookup^\acc$. 
\item %
The proof proceeds by structural induction. 
If $f$ is pure the result is obvious, otherwise 
$f$ can be written in a unique way as  
$f = u \circ \mathtt{op} \circ g$ where $u$ is pure, $\mathtt{op}$ 
is either $\lookup$ or $\update$ and $g$ is the remaining part of $f$.
By induction, either $g$ is an accessor or 
$g \eqs v\circ \lookup\circ \update \circ b$ 
for some pure term $v$ and some accessor $b$.
So, there are four cases to consider.
\begin{itemize}
\item If $\mathtt{op}=\lookup$ and $g$ is an accessor then 
$f$ is an accessor. 
\item If $\mathtt{op}=\update$ and $g$ is an accessor then 
by Point~\ref{pt:sta-canonical-acc} 
there is a pure term $w$ such that $u \eqs w\circ \lookup$,
so that $f \eqs w^\pure\circ \lookup\circ \update\circ g^\acc$. 
\item If $\mathtt{op}=\lookup$ and 
$g \eqs v^\pure\circ \lookup\circ \update \circ b^\acc$ then 
$f \eqs u \circ \lookup \circ v \circ \lookup\circ \update \circ b$. 
Since $v: V\to \unit$ is pure we have $v\circ \lookup \eqs \id_\unit$,
so that $f \eqs u^\pure \circ \lookup \circ \update \circ b^\acc$. 
\item If $\mathtt{op}=\update$ and 
$g \eqs v^\pure\circ \lookup\circ \update \circ b^\acc$ then 
$f \eqs u^\pure \circ \update \circ v^\pure\circ \lookup\circ 
\update \circ b^\acc$. 
Since $v$ is pure, by (ax) and (repl$_\eqw$) we have 
$ v\circ \lookup\circ \update \eqw v $. 
Besides, by (ax) and (subs$_\eqw$) we have 
$\lookup\circ \update\circ v \eqw v $ and 
$\lookup\circ \update\circ v \circ \lookup \circ \update \eqw 
v \circ \lookup \circ \update$. %
Since $\eqw$ is an equivalence relation these three weak equations imply 
$\lookup\circ \update\circ v\circ \lookup\circ \update \eqw 
\lookup\circ \update\circ v$. 
By rule (eq$_3$) we get 
$\update\circ v\circ \lookup\circ \update \eqs \update\circ v$, so that 
$f \eqs u^\pure \circ \update \circ (v \circ b)^\acc$.
\end{itemize}
\end{enumerate}
\end{proof}

Thanks to Proposition~\ref{prop:sta-canonical}, 
in order to study equations in the logic $L_\sta$
we may restrict our study to pure terms, 
accessors of the form $v^\pure\circ \lookup^\acc\circ \pa_X^\pure$ 
and modifiers of the form 
$u^\pure\circ \lookup^\acc\circ \update^\modi \circ a^\acc$. 

Point~\ref{pt:sta-equations-acc-pure} in Proposition~\ref{prop:sta-canonical} 
is \emph{not} dual to Point~\ref{pt:excore-canonical-acc} 
in Proposition~\ref{prop:excore-canonical}

\begin{proposition}
\label{prop:sta-equations} 
\begin{enumerate}
\item 
\label{prop:sta-equations-modi-modi} 
For all $a_1^\acc,a_2^\acc:X\to V$ and $u_1^\pure,u_2^\pure:V\to Y$, 
let $f_1^\modi = u_1\circ \lookup\circ \update\circ a_1:X\to Y$
and $f_2^\modi = u_2\circ \lookup\circ \update\circ a_2:X\to Y$, 
then
$ f_1 \eqw f_2 $ is $ T_\sta$-equivalent to 
  $ u_1\circ a_1 \eqs u_2\circ a_2 $ 
and 
$ f_1\eqs f_2 $ is $ T_\sta$-equivalent to 
  $ \{a_1\eqs a_2 \;,\; u_1\circ a_1 \eqs u_2\circ a_2 \}$. 
\item 
\label{prop:sta-equations-modi-acc}
For all $a_1^\acc:X\to V$, $u_1^\pure:V\to Y$ and $a_2^\acc:X\to Y$, 
let $ f_1^\modi = u_1\circ \lookup\circ \update \circ a_1:X\to Y$, 
then 
$ f_1 \eqw a_2 $ is $ T_\sta$-equivalent to 
  $ u_1\circ a_1 \eqs a_2 $ 
$ f_1 \eqs a_2 $ is $ T_\sta$-equivalent to 
  $ \{ u_1\circ a_1 \eqs a_2 \;,\; a_1\eqs \lookup\circ \pa_X\} $.
\item 
\label{pt:sta-equations-acc-acc} 
Let us assume that $\pa_X^\pure$ is an epimorphism with respect to accessors. 
For all $v_1^\pure,v_2^\pure:V\to Y$ 
let $a_1^\acc = v_1\circ \lookup\circ \pa_X:X\to Y$
and $a_2^\acc = v_2\circ \lookup\circ \pa_X:X\to Y$. 
Then 
$ a_1\eqs a_2 $ is $ T_\sta$-equivalent to 
  $ v_1\eqs v_2 $.
\item 
\label{pt:sta-equations-acc-pure}
Let us assume that $\pa_V^\pure$ is an epimorphism with respect to accessors
and that there exists a pure term $k_X^\pure:\unit\to X$.
For all $v_1^\pure:V\to Y$ and $v_2^\pure: X\to Y$, 
let $a_1^\acc = v_1 \circ \lookup \circ \pa_X: X\to Y$.
Then 
$ a_1 \eqs v_2 $ is $ T_\sta$-equivalent to 
  $ \{ v_1 \eqs v_2 \circ k_X \circ \pa_V \;,\; 
v_2 \eqs v_2 \circ k_X \circ \pa_X \} $. 
\end{enumerate}
\end{proposition}

\begin{proof} 
\begin{enumerate}
\item %
Rule (eq$_2$) implies that $f_1\eqs f_2$ if and only if 
$f_1\eqw f2$ and $\pa_Y \circ f_1 \eqs \pa_Y \circ f_2$. 
On the one hand, $f_1\eqw f_2$ if and only if $u_1\circ a_1 \eqs u_1\circ a_2$:
indeed, for each $i\in\{1,2\}$, by (ax) and (repl$_\eqw$), since $u_i$ is pure 
we have $f_i \eqw u_i\circ a_i$.
On the other hand, let us prove that 
$\pa_Y \circ f_1 \eqs \pa_Y \circ f_2$ if and only if $a_1 \eqs a_2$.
\begin{itemize}
\item For each $i\in\{1,2\}$, 
the accessor $\pa_Y \circ u_i\circ \lookup : \unit\to \unit$ satisfies  
$\pa_Y \circ u_i\circ \lookup \eqs \id_\unit$,
so that $\pa_Y \circ f_i \eqs \update\circ a_i$.
Thus, $\pa_Y \circ f_1 \eqs \pa_Y \circ f_2$ if and only if 
$\update\circ a_1 \eqs\update\circ a_2$.
\item Clearly, if $a_1 \eqs a_2$ then $\update\circ a_1 \eqs\update\circ a_2$.
Conversely, if $\update\circ a_1 \eqs\update\circ a_2$
then $\lookup\circ\update\circ a_1 \eqs \lookup\circ\update\circ a_2$,
so that by (ax) and (subs$_\eqw$) we get $a_1 \eqw a_2$,
which means that $a_1 \eqs a_2$ because $a_1$ and $a_2$ are accessors. 
\end{itemize}
\item %
Rule (eq$_2$) implies that $f_1\eqs a_2$ if and only if 
$f_1\eqw a_2$ and $\pa_Y \circ f_1 \eqs \pa_Y \circ a_2$. 
On the one hand, $f_1\eqw a_2$ if and only if $u_1\circ a_1 \eqs a_2$:
indeed, by (ax) and (repl$_\eqw$), since $u_1$ is pure 
we have $f_1 \eqw u_1\circ a_1$.
On the other hand, let us prove that 
$\pa_Y \circ f_1 \eqs \pa_Y \circ a_2$ if and only if 
$a_1 \eqs \lookup \circ \pa_X $, in two steps.
\begin{itemize}
\item 
Since $ \pa_Y \circ a_2 :X\to\unit $ is an accessor,
we have $ \pa_Y \circ a_2 \eqs \pa_X$. 
Since $ \pa_Y \circ f_1 = 
\pa_Y \circ u_1\circ \lookup\circ \update \circ a_1 $ 
with $\pa_Y \circ u_1\circ \lookup:\unit\to\unit $ an accessor,
we have $\pa_Y \circ u_1\circ \lookup \eqs \id_\unit$ 
and thus we get $ \pa_Y \circ f_1 \eqs \update \circ a_1 $. 
Thus, $\pa_Y \circ f_1 \eqs \pa_Y \circ a_2 $ if and only if 
$\update \circ a_1 \eqs \pa_X$. 
\item 
If $\update \circ a_1 \eqs \pa_X$ then 
$\lookup\circ\update \circ a_1 \eqs \lookup\circ\pa_X$, 
by (ax) and (subs$_\eqw$) this implies 
$a_1 \eqw \lookup\circ\pa_X$, which is a strong equality because 
both members are accessors. 
Conversely, if $a_1 \eqs \lookup \circ \pa_X $ 
then $\update \circ a_1 \eqs \update \circ \lookup \circ \pa_X $, 
by Point~\ref{pt:sta-ul-ul} in Lemma~\ref{lemm:sta-ul}  this implies 
$\update \circ a_1 \eqs \pa_X $. 
Thus, $\update \circ a_1 \eqs \pa_X$ if and only if 
$a_1 \eqs \lookup \circ \pa_X $. 
\end{itemize}
\item %
Clearly, if $v_1\eqs v_2$ then $a_1\eqs a_2 $. 
Conversely, if $a_1\eqs a_2 $, i.e., if 
$v_1\circ \lookup\circ \pa_X \eqs v_2\circ \lookup\circ \pa_X $, 
since  $\pa_X$ is an epimorphism with respect to accessors we get 
$v_1\circ \lookup \eqs v_2\circ \lookup $.
By Point~\ref{pt:sta-ul-lulu} in Lemma~\ref{lemm:sta-ul},
this means that $v_1 \eqs v_2$. 
\item %
Let $w_2^\pure = v_2 \circ k_X :\unit\to Y$. 
Let us assume that $v_1 \eqs w_2 \circ \pa_V$ and $v_2 \eqs w_2 \circ \pa_X$. 
Equation $v_1 \eqs w_2 \circ \pa_V$ implies 
$a_1 \eqs w_2 \circ \pa_V \circ \lookup \circ \pa_X$.
Since $ \pa_V \circ \lookup\eqs \id_\unit$ 
we get $a_1 \eqs w_2 \circ \pa_X$.
Then, equation $v_2 \eqs w_2 \circ \pa_X$ implies $a_1 \eqs v_2$. 
Conversely, let us assume that $a_1 \eqs v_2$,
which means that $v_1 \circ \lookup \circ \pa_X \eqs v_2$. 
Then $v_1 \circ \lookup \circ \pa_X \circ k_X \circ \pa_V 
\eqs v_2 \circ k_X \circ\pa_V$, 
which reduces to $v_1 \circ \lookup \circ \pa_V \eqs w_2 \circ\pa_V$.
Since $\pa_V$ is an epimorphism with respect to accessors we get 
$v_1 \circ \lookup \eqs w_2 $, which means that 
$v_1 \eqs w_2 \circ\pa_V$ by Point~\ref{pt:sta-ul-l} 
in Lemma~\ref{lemm:sta-ul}. 
Now let us come back to equation $v_1 \circ \lookup \circ \pa_X \eqs v_2$; 
since $v_1 \eqs w_2 \circ\pa_V$, it yields 
$w_2 \circ\pa_V \circ \lookup \circ \pa_X \eqs v_2$, 
so that $w_2 \circ\pa_X \eqs v_2$. 
\end{enumerate}
\end{proof}

The assumption for Theorem~\ref{theo:sta-complete} 
comes form the fact that the existence of a pure term $k_X^\pure:\unit\to X$, 
which is used in Point~\ref{pt:sta-equations-acc-pure} 
of Proposition~\ref{prop:sta-equations}, 
is incompatible with the intended model of states 
if $X$ is interpreted as the empty set. 
The assumption for Theorem~\ref{theo:sta-complete} 
is \emph{not} dual to the assumption for Theorem~\ref{theo:excore-complete}.

\begin{definition}
\label{defi:inhab} 
A type $X$ is \emph{inhabited} 
if there exists a pure term $k_X^\pure:\unit\to X$.
A type $\empt$ is \emph{empty} 
if for each type $Y$ there is a pure term $\copa_Y^\pure:\empt\to Y$,
and every term $f:\empt\to Y$ is such that $f\eqs\copa_Y$. 
\end{definition}

\begin{remark}
\label{rem:inhab} 
When $X$ is inhabited then for any $k_X^\pure:\unit\to X$ 
we have $\pa_X\circ k_X \eqs \id_\unit$, 
so that $\pa_X$ is a split epimorphism; 
it follows that $\pa_X$ is an epimorphism with respect to all terms, 
and especially with respect to accessors.  
\end{remark}

\begin{theorem}
\label{theo:sta-complete}
If every non-empty type is inhabited and if $V$ is non-empty, 
the theory of states $T_\sta$ is Hilbert-Post complete
with respect to the pure sublogic $L_\sta^\pure$ of $L_\sta$. 
\end{theorem}

\begin{proof}
Using Corollary~\ref{coro:hpc-equations},
the proof relies upon Propositions~\ref{prop:sta-canonical} and 
\ref{prop:sta-equations}. 
it follows the same lines as the proofs of Theorems~\ref{theo:exc-complete}
and~\ref{theo:excore-complete}. 
The theory $T_\sta$ is consistent:
it cannot be proved that $\update^\modi \eqs \pa_V^\pure$ 
because the logic $L_\sta$ is sound with respect to its intended model
and the interpretation of this equation in the intended model is false
as sson as $V$ has at least two elements: 
indeed, for each state $s$ and each $x\in V$, 
$\lookup \circ \update (x,s) = x$ because of (ax) 
while $\lookup\circ \pa_V (x,s) = \lookup(s)$ does not depend on $x$. 
Let us consider an equation (strong or weak) 
between terms with domain $X$ in $L_\sta$; we distinguish two cases, 
whether $X$ is empty or not. 
When $X$ is empty, then all terms from $X$ to $Y$ 
are strongly equivalent to $\copa_Y$, 
so that the given equation 
is $T_\sta$-equivalent to the empty set of equations between pure terms. 
When $X$ is non-empty then it is inhabited,
thus by Remark~\ref{rem:inhab} 
$\pa_X$ is an epimorphism with respect to accessors.  
Thus, Propositions~\ref{prop:sta-canonical} and~\ref{prop:sta-equations} 
prove that the given equation 
is $T_\sta$-equivalent to a finite set of equations between pure terms. 
\end{proof}

\begin{remark}
\label{rem:sta-complete}
This can be generalized to an arbitrary number of locations. 
The logic $L_\sta$ and the theory $T_\sta$ 
have to be generalized as in \cite{DDFR12-state}, 
then Proposition~\ref{prop:sta-canonical} has to be adapted 
using the basic properties of $\lookup$ and $\update$,    
as stated in  \cite{PlotkinPower02}; 
these properties can be deduced from the decorated theory for states,  
as proved in \cite{DDEP13-coq}. 
The rest of the proof generalizes accordingly, 
as in \cite{Pretnar10}. 
\end{remark}

\end{document}